\documentclass[12pt]{article}

\usepackage{amsmath,amsfonts, latexsym, amssymb, amsthm, mathtools}
\usepackage{natbib, setspace, lscape, graphicx, mathabx}

\usepackage[inner=2.5cm,outer=2.5cm,bottom=2.5cm]{geometry}

\usepackage{caption} 
\usepackage{subcaption}
\usepackage{psfrag,epsf,epstopdf}
\usepackage{enumerate}
\usepackage{enumitem}
\usepackage{tikz}
\usepackage{siunitx}
\usepackage{algorithm}
\usepackage{bm, bbm}    
\usepackage{threeparttable}   
\usepackage{tcolorbox}    
\usepackage{color,float}     
\usepackage{multirow}   
\usepackage[toc,page]{appendix}
\usepackage{hyperref}
\usepackage{hyperxmp}
\usepackage{orcidlink}

\DeclareMathOperator{\var}{var}               

\newcommand\independent{\protect\mathpalette{\protect\independenT}{\perp}}
\def\independenT#1#2{\mathrel{\rlap{$#1#2$}\mkern2mu{#1#2}}}

\usepackage{url}

\newtheorem{theorem}{Theorem}

\newtheorem{proposition}{Proposition} 
\newtheorem{assumption}{Assumption}

\renewcommand{\theassumption}{A\arabic{assumption}}
              
\def\1{1\!{\rm l}}

\def \R {\mathbb{R}}
\def \Rex {\overline{\mathbb{R}}}
\def \E {\mathbb{E}}

\def \X {\mathcal{X}}

\definecolor{darkblue}{rgb}{0.0, 0.0, 0.55}

\hypersetup{
  unicode     = true,
  pdftitle    = {Rectified Linear Unit Regression},
  pdfauthor   = {Tatsushi Oka},
  pdfsubject  = {ReLU regression and distributional treatment effects},
  pdfkeywords = {ReLU regression, 
                  distributional regression,
                  quantile treatment effects, 
                  convex duality},
  pdfcontacturl = {https://orcid.org/0000-0002-8868-6136},
  colorlinks   = true, 
  urlcolor     = darkblue, 
  linkcolor    = darkblue, 
  citecolor    = darkblue 
}

\graphicspath{{./figure/}}

\title{
  \textbf{Rectified Linear Unit Regression}\footnote{
    First version: June 25, 2024. This version: \today. 
    The author is grateful to 
    St\'ephane Bonhomme,
    Xiaohong Chen, 
    Iv\'an Fern\'andez-Val, 
    Bruce Hansen, 
    Chih-Sheng Hsieh, 
    Hiroyuki Kasahara, 
    Tetsuya Kaji, 
    Tong Li, 
    Chu-An Liu, 
    Ryo Okui,
    Taisuke Otsu, 
    Artem Prokhorov, 
    Shosei Sakaguchi,  
    Yinchu Zhu 
    for their valuable comments, which greatly improved this manuscript. 
    The author also thanks participants at
    SWET2024, 
    The World Congress of the Econometric Society 2025, 
    CRETA Seminar at National Taiwan University, Kansai Econometric Workshop 2026,
    the IAAE Annual Conference 2026, 
    the 2026 Joint Statistical Meeting for their suggestions.  
    The author gratefully acknowledges financial support from
    KDDI Research, Inc. 
  }
}
\author{Tatsushi Oka\footnote{ 
  Department of Economics, Keio University. 
  Email: \href{mailto:tatsushi.oka@keio.jp}{tatsushi.oka@keio.jp} 
  \orcidlink{0000-0002-8868-6136}
  }  
  }

\date{}

\begin{document}
 
\maketitle

\begin{abstract}

This paper develops a regression framework for analyzing integrated conditional distribution and quantile functions. 
The proposed method, termed rectified linear unit (ReLU) regression, projects the ReLU-transformed outcome onto covariates and admits a closed-form estimator. 
Its population regression function is the best linear approximation to the integrated conditional distribution function of the outcome,
and  
the corresponding convex conjugate, obtained via the Legendre-Fenchel transform, approximates the integrated conditional quantile function.
Both the regression and its conjugate require only mild distributional assumptions and accommodate non-continuous outcomes. We establish the asymptotic distribution of the estimator and develop inference for the conjugate functional via the delta method for Hadamard directionally differentiable maps.
Building on these results, we establish identification and inference for average quantile treatment effects over arbitrary subintervals of probability levels.
This broadens the set of distributional parameters available to empirical work.

\end{abstract}

\vspace{0.5cm}


\setcounter{section}{0}

\section{Introduction}

Regression is a cornerstone of data analysis. While mean regression has long been the standard, modern approaches such as quantile regression \citep{koenker1978regression}, expectile regression \citep{newey1987asymmetric}, and distributional regression \citep{foresi1995conditional, chernozhukov2013inference} offer a more comprehensive view of the underlying process. 
These methods capture distributional heterogeneity that the conditional mean alone overlooks.
Recovering this heterogeneity is often what matters most for causal inference and policy evaluation.

This paper introduces an approach termed rectified linear unit (ReLU) regression, which regresses the ReLU transform $\max\{0, y - Y\}$ of the outcome $Y$ on covariates at each threshold $y$.
The ReLU transformation was first explored in neural computation research \citep{hahnloser2000digital} and later gained prominence in deep learning through the work of \cite{nair2010rectified} on restricted Boltzmann machines. 
A primary advantage of this approach is that ReLU regression captures distributional features under mild conditions while admitting a closed-form estimator. 
Through the Legendre-Fenchel transform, the regression yields 
a convex-duality representation that approximates the integrated conditional quantile function.
This convex structure allows for a direct characterization of distributional treatment effects.

The ReLU regression model offers three distinct advantages in analyzing distributional features. 
First, the model admits a closed-form estimator by formulating the estimation as an $L_2$ minimum distance problem with the ReLU-transformed dependent variable.
Quantile regression \citep{koenker1978regression}, expectile regression \citep{newey1987asymmetric}, and distributional regression \citep{foresi1995conditional}, by contrast, typically require iterative optimization. The closed-form structure simplifies both computation and the derivation of asymptotic properties.

Next, the ReLU regression model imposes weaker conditions on the data-generating process, in that only finite moments and standard rank conditions are required.
Conversely, standard asymptotic theory in quantile regression hinges on a strictly positive conditional density at the quantile of interest together with additional smoothness conditions \citep[see, e.g.,][]{koenker2005quantile,
angrist2006quantile}, and the limiting distribution may become non-standard in the absence of these regularity requirements \citep{smirnov1952limit, knight2002limiting}.
The distinction matters most for non-continuous outcomes, including discrete, count, and mixed discrete-continuous variables.
In such settings, \cite{machado2005quantiles} introduce a jittering-based estimation method, and \cite{chernozhukov2020generic} develop generic distributional inference. ReLU regression instead targets an integrated functional, which smooths over the jumps and flat regions of non-continuous distributions.

Finally, our methodology supports a unified treatment of
distributional treatment effects through convex duality. 
By approximating the integrated conditional distribution function across arbitrary locations, the ReLU regression model provides a natural
complement to the conditional value-at-risk paradigm of
\cite{rockafellar2000optimization, rockafellar2002conditional}.
Applying the Legendre-Fenchel transform to the fitted regression functions approximates the integrated conditional quantile function, which connects ReLU regression to the stochastic dominance principles of \cite{ogryczak2002dual} and the coherent risk theory of
\cite{acerbi2002coherence}.
Under a saturated specification, the projection and its conjugate coincide with the integrated conditional distribution and quantile functions.
The resulting framework quantifies average quantile treatment effects (AQTE)  over arbitrary subintervals of quantile levels
and includes both the average treatment effect (ATE) and the
quantile treatment effect (QTE) as special cases, building on
\cite{firpo2007efficient} and \cite{chernozhukov2013inference}.

\subsection{Related Literature}

Our framework connects to several strands of literature beyond
those already mentioned.

First, our paper contributes to a growing literature on
heterogeneous treatment effects in econometrics and statistics.
The quantile treatment effect (QTE) was first introduced by
\cite{doksum1974empirical} and \cite{lehmann1975nonparametrics} as
a measure of treatment effects under unobserved heterogeneity.
Since then, a large body of work has developed identification,
estimation, and inference methods for distributional and quantile
treatment effects, including \cite{heckman1997making},
\cite{athey2006identification}, \cite{bitler2006mean},
\cite{djebbari2008heterogeneous}, \cite{donald2014estimation},
\cite{callaway2018quantile, callaway2019quantile}, and
\cite{firpo2007efficient, firpo2009unconditional}, among others.
For discrete outcomes, \cite{chernozhukov2020generic} develop
generic distributional inference for treatment effects. Our
framework instead routes inference through the convex-duality
structure of the integrated conditional distribution function and
addresses non-continuous outcomes without smoothing or jittering
the underlying distribution. To our knowledge, however, the
average quantile treatment effect over an arbitrary subinterval
has not been studied as a unified parameter that nests both the
average treatment effect and the quantile treatment effect,
and remains point-identified for non-continuous outcomes.

The second strand of literature studies the intersection of convex
analysis and risk measurement, in which expected shortfall, also
known as conditional value-at-risk (CVaR) or superquantiles, links
risk measure theory and stochastic optimization. The foundations
of coherent risk measures were established by
\cite{artzner1999coherent}, followed by the mathematical framework
of \cite{rockafellar2000optimization, rockafellar2002conditional}.
\cite{pflug2000some} developed convexity properties and
computational approaches, while \cite{acerbi2002coherence} 
proved the coherence of expected shortfall.
\cite{follmer2016stochastic} provide  
a comprehensive treatment of risk measures, including convex and monetary risk measures. The
duality between conjugate functions and integrated quantile
functions for unconditional distributions, established by
\cite{ogryczak2002dual} and \cite{rockafellar2000optimization},
motivates our extension to the conditional case. 
The standard two-step approach for estimating the integrated quantile function first estimates the quantile function pointwise and then integrates it \citep[see, e.g.,][]{chen2025estimation},
which typically requires the outcome to admit a conditional density.
\citet{kaji2019asymptotic} studies an extension of $L$-statistics in which transformed empirical quantile functions are
integrated against random weight functions.
We instead estimate the best linear approximation to the integrated conditional distribution function directly by ReLU regression and apply the Legendre-Fenchel transform, thereby accommodating non-continuous outcomes.

The third strand of literature concerns inference for functionals
that are not fully Hadamard differentiable. The standard delta
method requires full Hadamard differentiability of the functional,
a condition that fails for the Legendre-Fenchel transform.
Foundational results on Hadamard directional differentiability are
due to \cite{shapiro1990} and \cite{dumbgen1993}.
\cite{fang2019inference} develop the delta method for Hadamard
directionally differentiable maps and characterize the
inconsistency of the standard nonparametric bootstrap in this
setting. Inference for shape-constrained or
directionally differentiable functionals also appears in
\cite{delgado2012testing}, \cite{beare2015transforming},
\cite{chernozhukov2010estimation}, and \cite{chen2021shape}, among
others. We apply this machinery to the conjugate functional in the
ReLU regression framework.

\subsection{Outline}

Section \ref{sec:relu} introduces the ReLU regression model and
its population interpretation through convex duality. Section
\ref{sec:estimation} develops estimation and the asymptotic
distribution of the proposed estimator, with inference for the conjugate
functional. Section \ref{sec:treatment} applies the framework to causal inference.
Section \ref{sec:application} reports an empirical application to
the Oregon Health Insurance Experiment, and Section
\ref{sec:conclusion} concludes. 

\section{ReLU Regression}
\label{sec:relu}

In this section, we introduce the ReLU regression model and provide its distributional interpretation through convex duality.

\textit{Notation}.
Throughout the paper, 
we use $\|\cdot\|$ to denote the Euclidean norm for vectors and the spectral norm for matrices, and the superscript $\top$ denotes transposition. For an arbitrary index set $T$, $\ell^{\infty}(T)$ denotes the space of uniformly bounded real-valued functions on $T$, and we write $X_{n}\rightsquigarrow X$ for weak convergence of $X_{n}$ to $X$ in $\ell^{\infty}(T)$. We also write $\1\{A\}$ for the indicator function, which equals $1$ if the event $A$ occurs and $0$ otherwise. 
We denote the real line by $\R$
and define the half-line
$\R_{+}:=[0,\infty)$ and
the extended real line
$\Rex:=[-\infty,+\infty]$.
For a proper convex function $f:\R\to\R\cup\{+\infty\}$, the subdifferential of $f$ at $x_{0}$ is defined by 
$\partial f(x_{0}):=
  \{\delta\in\R: 
  f(x)\geq f(x_{0})+\delta(x-x_{0}) \text{ for all } x \in\R \}
  $.

\subsection{Data and ReLU Regression Model}

Let $Y$ be a scalar outcome taking values in $\R$, and let
$X \in \X \subseteq \R^{p}$ be a vector of covariates
containing a constant term. 
The conditional distribution function of $Y$ given $X=x$
is defined as 
$F_{Y | X}(y | x) := \Pr\{Y \le y | X=x\}$
for
$y \in \R$,
and 
the corresponding conditional quantile function is  
$F_{Y | X}^{-1}(u | x) := \inf\{y \in \R : F_{Y | X}(y | x) \ge u\}$
for $u \in (0,1)$.

We introduce the ReLU regression model, 
whose key building block
is the ReLU transformation 
$(a)_{+} := \max\{0,a\}$ for $a \in \R$.
For each fixed $y \in \R$, 
we treat the ReLU-transformed outcome  
$(y-Y)_{+}$ 
as the dependent variable in the regression  
\begin{equation}
  \label{eq:relu-reg}
  (y - Y)_{+} = X^{\top} \beta_{0}(y) + \epsilon(y),
\end{equation}
where the coefficient vector $\beta_{0}(y) \in \R^{p}$ and the error $\epsilon(y)$ are
both indexed by $y$.
We define $\beta_{0}(y)$ as a minimizer over $\beta\in\R^{p}$ of the population least-squares criterion
\begin{equation*}
  Q(\beta; y) := 
  \E\Big[ \big( (y - Y)_{+} - X^{\top} \beta \big)^{2} \Big].
\end{equation*}
Throughout the paper, $X$ denotes a
finite-dimensional regressor vector, which may include transformations
such as polynomials, splines, and interactions.

We impose the following regularity conditions on the joint
distribution of $(X, Y)$.

\vspace{0.2cm}
\begin{assumption}
  \label{assump:dgp}
  The data-generating process satisfies:
  \begin{enumerate}[label=(\alph*), noitemsep, topsep=0pt]
    \item[\textnormal{(a)}]
    $\E[Y^2] < \infty$ and $\E[\|X\|^{2}] < \infty$.
    \item[\textnormal{(b)}]
    The matrix $\E[X X^{\top}]$ is positive definite
    and 
    the first element of $X$ equals one.
  \end{enumerate}
\end{assumption}
\vspace{0.2cm}

Assumption \ref{assump:dgp} is the standard regularity condition for
the population $L^{2}$ problem. 
Assumption \ref{assump:dgp}(a) ensures that the
ReLU-transformed outcome $(y-Y)_{+}$ and any
linear combination $X^{\top}\beta$ are square integrable for each fixed $y \in \R$.
Assumption \ref{assump:dgp}(b) imposes the full-rank condition on the Gram matrix of $X$
and 
requires the regressor to include a constant.

Under Assumption \ref{assump:dgp}, the criterion
$Q(\beta;y)$ is finite and strictly convex in $\beta$ for each
$y\in\R$. Hence, $\beta_{0}(y)$ is uniquely defined and admits the
closed-form expression
\begin{equation*}
  \beta_{0}(y)
  =
  \big(\E[XX^{\top}]\big)^{-1}
  \E[X(y-Y)_{+}].
\end{equation*}
The corresponding first-order condition is
$\E[X\epsilon(y)]=0$ for every $y\in\R$.
The linear predictor $X^{\top}\beta_{0}(y)$ is the best linear
predictor of $(y - Y)_{+}$ given $X$ in the $L^{2}$ sense. 
This interpretation
does not require the conditional expectation
$\E[(y-Y)_{+}| X]$ to be linear in $X$. When it is nonlinear,
$\beta_{0}(y)$ has the usual pseudo-true interpretation as a population
projection coefficient. This interpretation is analogous to those of
linear mean regression \citep{white1980using} and linear quantile
regression \citep{angrist2006quantile} under misspecification.
Proposition A1 in Appendix A states this result
formally.

\subsection{Convex Duality and Distributional Interpretation}

We next develop the distributional interpretation of ReLU regression
through convex duality. We first establish the duality between the
integrated conditional distribution and quantile functions and then
relate these functions to ReLU regression.

For any $y \in \R$, the integrated conditional distribution function
of $Y$
given $X=x$
is defined as
\begin{equation*}
  G_{Y | X}(y | x) := \int_{-\infty}^{y} F_{Y | X}(s | x)\, ds.
\end{equation*}
The unconditional analogue has been studied in the risk-measurement
literature \citep{follmer2016stochastic}. The map
$y \mapsto G_{Y | X}(y | x)$ is convex for each fixed $x \in \X$
since the integrand $F_{Y | X}(\cdot | x)$ is nondecreasing.
Convexity allows us to apply the Legendre-Fenchel transform. For
each $\tau \in [0,1]$ and $x \in \X$, the conjugate function is
defined as
\begin{equation*}
  G_{Y | X}^{\ast}(\tau | x)
  :=
  \sup_{y \in \R}\big\{ \tau y - G_{Y | X}(y | x) \big\}.
\end{equation*}

The next proposition characterizes the convexity and the
subdifferential structure of $G_{Y | X}(\cdot | x)$ and its conjugate
$G_{Y | X}^{\ast}(\cdot | x)$.
We write 
$F_{Y | X}(y {\scalebox{0.7}{$-$}} | x) := \lim_{s \nearrow y}
F_{Y | X}(s | x)$ for the left limit of the conditional distribution
function and $F_{Y | X}^{-1}(\tau{\scalebox{0.7}{$+$}}
| x) := \lim_{u \searrow \tau} F_{Y | X}^{-1}(u | x)$ for the
right limit of the conditional quantile function.

\vspace{0.5cm}
\begin{proposition}
  \label{proposition:L2-min-1}
  Suppose that $\E[|Y|] < \infty$. 
  Then, for almost every $x \in \X$,
  the following statements hold.
  \begin{enumerate}[label=(\alph*), itemsep=0.2em]
    \item The map $y \mapsto G_{Y | X}(y | x)$ is convex, and for
      any $y \in \R$,
      \begin{equation*}
        G_{Y | X}(y | x) = \E[(y-Y)_{+} | X=x]
      \end{equation*}
      Moreover, $\partial G_{Y | X}(y | x) =
      [F_{Y | X}(y {\scalebox{0.7}{$-$}} | x),\, F_{Y | X}(y | x)]$ for any
      $y \in \R$.
    \item The conjugate $\tau \mapsto G_{Y | X}^{\ast}(\tau | x)$
      is convex, and for any $\tau \in [0,1]$,
      \begin{equation*}
        G_{Y | X}^{\ast}(\tau | x)
        =
        \int_{0}^{\tau} F_{Y | X}^{-1}(u | x)\, du.
      \end{equation*}
      Correspondingly, $\partial G_{Y | X}^{\ast}(\tau | x) =
      [F_{Y | X}^{-1}(\tau | x),\,
      F_{Y | X}^{-1}(\tau{\scalebox{0.7}{$+$}} | x)]$
      for any $\tau \in (0,1)$.
  \end{enumerate}
\end{proposition}
\vspace{0.5cm}

Proposition \ref{proposition:L2-min-1} establishes the relationship
between the ReLU transformation and the integrated quantile function
through convex duality. The unconditional counterpart of this
relationship is well known in the risk measurement literature
\citep{ogryczak2002dual, rockafellar2000optimization,
rockafellar2002conditional}. The proof of the proposition is given
in the appendix. The duality structure is summarized in Figure
\ref{fig:duality_structure}.

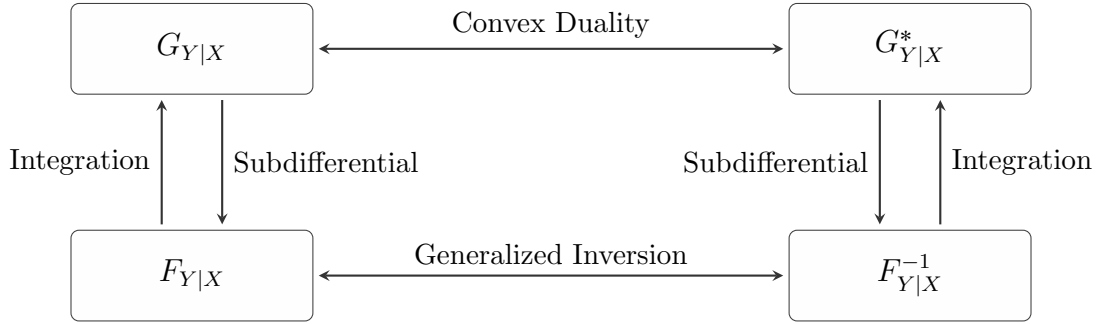
\begin{figure}[htbp]
\centering
\caption{Duality of the Integrated Distribution and Quantile Functions}
\label{fig:duality_structure}
\vspace{0.1cm}

\begin{tikzpicture}[
  node distance = 3.6cm,
  box/.style = {draw=black!75, fill=gray!1, rounded corners=3pt,
    rectangle, minimum width=3.2cm, minimum height=1.2cm,
    align=center},
  arr/.style = {->, >=stealth, thick, draw=black!80,
    shorten >=2pt, shorten <=2pt},
  darr/.style = {<->, >=stealth, thick, draw=black!80,
    shorten >=2pt, shorten <=2pt}
]

\node[box] (G)  at (0, 0)    {$G_{Y | X}$};
\node[box] (Gs) at (9.5, 0)  {$G_{Y | X}^{*}$};
\node[box] (F)  at (0, -3)   {$F_{Y | X}$};
\node[box] (Fi) at (9.5, -3) {$F_{Y | X}^{-1}$};

\draw[darr] (G) -- (Gs)
  node[midway, above, font=\small] {Convex Duality};

\draw[darr] (F) -- (Fi)
  node[midway, above, font=\small] {Generalized Inversion};

\draw[arr] ([xshift=-0.4cm]F.north) -- ([xshift=-0.4cm]G.south)
  node[midway, left, font=\small] {Integration};
\draw[arr] ([xshift=0.4cm]G.south) -- ([xshift=0.4cm]F.north)
  node[midway, right, font=\small] {Subdifferential};

\draw[arr] ([xshift=0.4cm]Fi.north) -- ([xshift=0.4cm]Gs.south)
  node[midway, right, font=\small] {Integration};
\draw[arr] ([xshift=-0.4cm]Gs.south) -- ([xshift=-0.4cm]Fi.north)
  node[midway, left, font=\small] {Subdifferential};

\end{tikzpicture}

\vspace{0.5cm}
\begin{minipage}{0.93\textwidth}
\small
\textit{Note}: The top row illustrates the convex duality between the
integrated functions, while the bottom row reflects the relationship
between the distribution and quantile functions via generalized
inversion.
\end{minipage}
\end{figure}

Proposition \ref{proposition:L2-min-1} has three implications. First,
the result requires only the finite first-moment condition $\E[|Y|]
< \infty$ and is agnostic about the existence or smoothness of
conditional densities. Second, the integrated functions $G_{Y | X}(y |
x)$ and $G_{Y | X}^{\ast}(\tau | x)$ are well-defined for arbitrary
outcome distributions, including discrete and mixed
discrete-continuous cases. 
Although the generalized inverse uniquely defines a  conditional quantile function, it can be an ill-posed target of inference at probability indices where the distribution function is flat.
At such indices, the inverse map from the distribution function
to its quantile is discontinuous. 
By contrast, the integrated quantile function is invariant to the selection from the quantile correspondence. Any two
  selections differ only on a Lebesgue-null set of probability indices
  and therefore yield the same integral.
Third, the integrated functions fully characterize the
underlying distribution. The conditional distribution function and
the conditional quantile function are recovered as elements of the
subdifferentials, $F_{Y | X}(y | x) \in \partial G_{Y | X}(y | x)$ and
$F_{Y | X}^{-1}(\tau | x) \in \partial G_{Y | X}^{\ast}(\tau | x)$. 
For distributions with positive density, these subdifferentials reduce to singletons, whereas for non-smooth distributions they remain well-defined as set-valued maps.

To connect the preceding duality result with ReLU regression, define
the ReLU projection function
and its associated Legendre-Fenchel conjugate as 
\begin{equation*}
  \widetilde{G}_{Y | X}(y | x)
  :=
  x^{\top}\beta_{0}(y)
  \qquad \text{and} \qquad 
  \widetilde{G}_{Y | X}^{\ast}(\tau | x)
  :=
  \sup_{y\in\R}
  \big\{\tau y- \widetilde{G}_{Y | X}(y | x)\big\}.
\end{equation*}
Both functions are determined by 
the pseudo-true coefficient function 
$\beta_{0}(\cdot)$.
  The ReLU projection has a direct distributional interpretation when
  the conditional ReLU mean is exactly represented by the linear
  specification. 
Specifically, if 
  $
    \E[(y-Y)_{+} | X]
    =
    X^{\top}\beta_{0}(y)
  $
  a.s.~for every $y \in\R$,
then Proposition \ref{proposition:L2-min-1}(a) implies that
  $\widetilde{G}_{Y| X}(\cdot |  x)
  =G_{Y| X}(\cdot | x)$, and thus
  $\widetilde{G}_{Y| X}^{\ast}(\cdot | x)
  =G_{Y| X}^{\ast}(\cdot | x)$. 
The condition holds when the conditioning variables have finite
support and a saturated specification is used.
This includes the treatment-group specification
  in Section \ref{sec:treatment} and the fully interacted
  saturated-stratum specification used in Section
  \ref{sec:application}. In more parsimonious
  specifications, $\widetilde{G}_{Y | X}$ retains its
  best-linear-approximation interpretation, with
  $\widetilde{G}_{Y| X}^{\ast}$ as its associated convex-dual
  functional.

\section{Estimation and Asymptotic Properties}
\label{sec:estimation}

In this section, we present the ReLU regression estimator and derive its asymptotic distribution. 
We then develop inference for the convex conjugate of the ReLU projection.

\subsection{Estimation Method}

We observe 
an independent and identically distributed (i.i.d.) sample 
$\{(X_{i}, Y_{i})\}_{i=1}^{n}$ drawn from the joint distribution of
$(X, Y)$.
For each fixed $y \in \R$, we define the sample objective function
as
\begin{equation*}
  \widehat{Q}(\beta; y)
  :=
  \frac{1}{n} \sum_{i=1}^{n}
  \big( (y - Y_{i})_{+} - X_{i}^{\top} \beta \big)^{2}.
\end{equation*}
We define the estimator $\hat{\beta}(y)$ as the minimizer of
$\widehat{Q}(\beta; y)$ over $\beta \in \R^{p}$. 
Under Assumption \ref{assump:dgp}, the sample second-moment matrix
$n^{-1} \sum_{i=1}^{n} X_{i} X_{i}^{\top}$ is invertible with
probability approaching one. The estimator then admits the
closed-form expression
\begin{equation*}
  \hat{\beta}(y)
  =
  \bigg( \sum_{i=1}^{n} X_{i} X_{i}^{\top} \bigg)^{-1}
  \sum_{i=1}^{n} X_{i} (y - Y_{i})_{+}.
\end{equation*}

Given the estimator $\hat{\beta}(y)$, we define an estimator
of the ReLU projection function
$\widetilde{G}_{Y | X}(\cdot|x)$ for each $x \in \X$  by
\begin{equation*}
  \widehat{G}_{Y | X}(y | x)
  :=
  x^{\top} \hat{\beta}(y).
\end{equation*}
Evaluating this estimator across $y \in \R$, we obtain the collection
$\{\widehat{G}_{Y | X}(y | x) : y \in \R\}$.
We estimate the convex-dual functional associated with the ReLU projection  by applying
the Legendre-Fenchel transform to $\widehat{G}_{Y | X}(\cdot | x)$.
For each fixed $x \in \X$ and $\tau \in [0, 1]$, the conjugate
estimator is
\begin{equation*}
  \widehat{G}_{Y | X}^{\ast}(\tau | x)
  :=
  \sup_{y \in \R}
  \big\{ \tau y - \widehat{G}_{Y | X}(y | x) \big\}.
\end{equation*}
The estimated ReLU projection map  
$y \mapsto \widehat{G}_{Y | X}(y | x)$ need not be convex.
Applying the Legendre-Fenchel transform once more to
  $\widehat{G}_{Y | X}^{\ast}(\cdot | x)$,
we can obtain the greatest convex minorant of $\widehat{G}_{Y | X}(\cdot | x)$.

\subsection{Asymptotic Properties}

We now derive the asymptotic distribution of our proposed estimator.
To this end, we impose the following regularity conditions.

\vspace{0.3cm}
\begin{assumption}
  \label{assump:regularity}
  The following conditions hold:
  \begin{enumerate}[noitemsep, topsep=0pt]
    \item[\textnormal{(a)}]
    $\{(X_{i}, Y_{i})\}_{i=1}^{n}$ is an i.i.d.\ sample from the distribution of
    $(X, Y)$.
    \item[\textnormal{(b)}]
    $\E[ \| X Y \|^{2}] < \infty$ and $\E[ \|X\|^{4} ] < \infty$.
    \end{enumerate}
\end{assumption}
\vspace{0.3cm}

Assumption \ref{assump:regularity}(a) requires random
sampling.
Assumption \ref{assump:regularity}(b) imposes mild moment conditions,
which 
hold if $\E[Y^{4}] < \infty$ and $\E[\|X\|^{4}] < \infty$, and
which imply Assumption \ref{assump:dgp}(a) because $X$ contains a constant.

The next theorem establishes weak convergence of
the centered coefficient process
$\sqrt{n} \big ( \hat{\beta}(\cdot) -  \beta_{0}(\cdot) \big )$ 
to a Gaussian process in $\ell^{\infty}(\Rex)^p$. 
We extend the process from $\R$ to $\Rex$ by adjoining its limits
as $y \to \pm\infty$, whose existence is established in Appendix A.
Also, we endow $\Rex$ with the metric
$\rho(y_{1},y_{2}):=|\arctan(y_{1})-\arctan(y_{2})|$,
where 
$\arctan(\pm \infty)=\pm\pi/2$.
This metric agrees with the usual topology on $\R$ and makes
$(\Rex,\rho)$ a compact metric space.

\vspace{0.2cm}
\begin{theorem}
  \label{theorem:aym}
  Suppose that Assumptions \ref{assump:dgp} and
  \ref{assump:regularity} hold. Then
  \begin{equation*}
    \sqrt{n}\big( \hat{\beta}(\cdot) - \beta_{0}(\cdot) \big)
    \rightsquigarrow
    \mathbb{B}(\cdot)
    \quad \text{in } \ell^{\infty}(\Rex)^{p}, 
  \end{equation*}
  where $\mathbb{B}(\cdot)$ is a zero-mean Gaussian process with
  uniformly continuous sample paths with
  respect to $\rho(\cdot, \cdot)$ and covariance function
  $Q_{X}^{-1}\, \Sigma(y_{1}, y_{2})\, Q_{X}^{-1}$ for
  $(y_{1}, y_{2}) \in \Rex^{2}$, with $Q_{X} := \E[X X^{\top}]$ 
  and
  $\Sigma(y_{1}, y_{2}) := 
  \mathrm{Cov} \big( X \epsilon(y_{1}), X \epsilon(y_{2}) \big)$. For each fixed $x \in \X$,
  \begin{equation*}
    \sqrt{n}\big(
      \widehat{G}_{Y | X}(\cdot | x) 
      - 
      \widetilde{G}_{Y | X}(\cdot | x)
    \big)
    \rightsquigarrow
    x^{\top} \mathbb{B}(\cdot)
    \quad \text{in } \ell^{\infty}(\Rex). 
  \end{equation*}
 
\end{theorem}
\vspace{0.3cm}

Uniform continuity of the sample paths of $\mathbb{B}(\cdot)$ supports
the construction of uniform confidence bands for
$\widetilde{G}_{Y | X}(\cdot | x)$ over $\R$. 
The covariance kernel has the heteroskedasticity-robust sandwich form.

The limiting distribution in Theorem \ref{theorem:aym} depends on
unknown nuisance parameters and is not pivotal. To obtain practical
inference for $\widehat{G}_{Y | X}(y | x)$, we apply the exchangeable
bootstrap \citep{praestgaard1993exchangeably}, which consistently
estimates the limit law of the linear functional
$x^{\top} \mathbb{B}(\cdot)$. When the data are organized into independent clusters, the same
bootstrap applies with weights drawn once per cluster and assigned
to every observation in that cluster. Appendix C states the cluster
extension and provides references for its validity.

We next study the asymptotic behavior of the conjugate estimator
$\widehat{G}_{Y | X}^{\ast}(\tau | x)$. The Legendre-Fenchel
transformation is not fully Hadamard differentiable. Thus, the standard
functional delta method of \cite{van1996weak} does not apply. We
work instead with Hadamard directional differentiability, which
extends the delta method to maps that are differentiable only in
selected directions \citep{shapiro1990, dumbgen1993, fang2019inference}.

Theorem 2.1 of \cite{carcamo2020} establishes Hadamard directional
differentiability of the supremum operator on the bounded function space $\ell^\infty(\R)$. 
In our case, however, the map
$y \mapsto\tau y-\widetilde{G}_{Y| X}(y | x)$ 
inside the Legendre-Fenchel transform may be unbounded below
and does not belong to $\ell^\infty(\R)$.
Moreover, the supremum may not be attained at a finite value of $y$ because maximizing sequences can diverge to $-\infty$ at $\tau=0$ or to $+\infty$ at $\tau=1$.
To address these issues, we define the set of approximate maximizers for each $\eta > 0$ by 
\begin{equation*}
  S_{\eta}(\tau,x)
  :=
  \big \{
    y\in\mathbb R:
  \tau y-\widetilde G_{Y| X}(y | x)
  \ge
  \widetilde G_{Y| X}^{*}(\tau | x)-\eta
  \big \}.
\end{equation*}
Given that $\widetilde{G}_{Y|X}^*(\tau|x)$ is finite, 
the set $S_{\eta}(\tau,x)$ is nonempty for every $\eta>0$, which ensures that the directional derivative depends only on $y \in S_{\eta}(\tau,x)$ and is well-defined.
Even when sequences of approximate maximizers in $S_{\eta}(\tau,x)$ diverge to $\pm \infty$, continuous extension to $\Rex$ yields well-defined boundary values. The boundedness of these extended functions and the nested structure of the collection $\{S_{\eta}(\tau,x)\}_{\eta > 0}$ then ensure that the relevant supremum converges to a finite limit as $\eta \downarrow 0$.

The following theorem characterizes the joint asymptotic
distribution of the conjugate estimator at any finite collection of
probability indices through the limit of
the suprema over these sets.
We also consider the special case in which 
the map 
$\tau \mapsto \widetilde{G}_{Y| X}^{\ast}(\tau | x)$ 
is differentiable, with the derivative 
$\nabla\widetilde{G}_{Y| X}^{\ast}(\tau | x)$,
and 
obtain weak convergence to a Gaussian limit.
Here, 
the derivatives at $\tau =0$ and $\tau =1$ are understood as the right and left derivatives, respectively.

\vspace{0.3cm}
\begin{theorem}
    \label{theorem:L-weak}
    Suppose that Assumptions 
    \ref{assump:dgp}-\ref{assump:regularity} hold.  
    Then the following statements hold.
  \begin{enumerate}[label=(\alph*), itemsep=0.2em]
    \item For every $x \in \X$
    and 
    for every finite collection
    $\{\tau_{j}\}_{j=1}^{J}\subset [0,1]$,
    \begin{equation*}
      \big \{ 
        \sqrt{n}
        \big ( 
          \widehat{G}_{Y | X}^{\ast}(\tau_j | x)
          -
          \widetilde{G}_{Y | X}^{\ast}(\tau_j | x)
        \big )
      \big \}_{j=1}^{J}
      \rightsquigarrow
      \big \{ 
      \mathbb{Z}(\tau_{j} | x)
      \big \}_{j=1}^{J},
    \end{equation*}
    where 
    $\mathbb{Z}(\tau|x)
    :=
    \lim_{\eta \downarrow0}
    \sup_{y \in S_{\eta}(\tau,x)}
    \{-x^{\top}\mathbb{B}(y)\}$
    for $\tau \in [0,1]$
    with 
    $\mathbb{B}(\cdot)$ being the Gaussian process defined in
    Theorem \ref{theorem:aym}.

    \item 
    Additionally, if 
    there exists a closed set $T\subseteq[0,1]$
    such that 
    $\widetilde{G}_{Y| X}^{\ast}(\cdot | x)$ 
    is differentiable on $T$
    and 
    $
    \lim_{\eta\downarrow0}
    \sup_{\tau\in T}
    \sup_{y\in S_{\eta}(\tau,x)}
    \rho \big(
    y,
    \nabla\widetilde{G}_{Y| X}^{\ast}(\tau | x)
    \big)
    =0,
    $
    then 
    the result in (a) strengthens to weak convergence in
    $\ell^{\infty}(T)$, with the Gaussian limit  
    $-x^{\top}\mathbb{B} 
    \big ( 
      \nabla \widetilde{G}_{Y| X}^{\ast}(\cdot|x) 
    \big )$. 

  \end{enumerate}
  
\end{theorem}
\vspace{0.3cm}

Under the differentiability condition in the theorem, the
Legendre-Fenchel transform is Hadamard differentiable, the limit is
a Gaussian process, and the standard nonparametric bootstrap is
consistent. In general, however, the theorem establishes only
finite-dimensional convergence to a non-pivotal limit that depends
on the unknown approximate-maximizer sets.
Moreover, the supremum over these sets generally yields a nonlinear
directional derivative, which makes the standard nonparametric
bootstrap inconsistent
\citep{dumbgen1993, fang2019inference}.
We thus adopt the resampling approach of \cite{fang2019inference} for
  Hadamard directionally differentiable maps, while
  \cite{delgado2012testing} and \cite{beare2015transforming}, among
  others, provide alternative approaches to shape-constrained inference.

To estimate the approximate-maximizer sets appearing in 
Theorem \ref{theorem:L-weak}, we follow the fattened-argmax approach of
\cite{chernozhukov2013intersection} and
\cite{cattaneo2020bootstrap}. 
Specifically, we use the empirical level
set
\begin{equation*}
  \widehat{S}_{\eta_{n}}(\tau,x)
  :=
  \big \{
    y\in\mathbb R:
  \tau y-\widehat{G}_{Y| X}(y | x)
  \ge
  \widehat{G}_{Y| X}^{*}(\tau | x)-\eta_{n}
  \big \},
\end{equation*}
where $\{\eta_n\}$ is a non-random sequence satisfying
$\eta_n \downarrow 0$ and
$\sqrt{n}\eta_n \to \infty$
as $n \to \infty$.

\section{Treatment Effect}
\label{sec:treatment}

\subsection{Setup and the Treatment Effects}

We extend the ReLU regression to evaluate distributional treatment
effects in a binary treatment setting under the potential
outcomes framework \citep{neyman1923applications,
rubin1974estimating}. Let $W \in \{0, 1\}$ denote the binary
treatment indicator, with $W = 1$ for treated units and $W = 0$ for control units, and let $Y(0)$ and $Y(1)$ denote the corresponding
potential outcomes. Under the stable unit treatment value
assumption \citep{rubin1980randomization}, the observed outcome is
$Y = Y(W)$.
For each $w \in \{0, 1\}$, 
let 
$F_{Y(w)}(y) := \Pr\{Y(w) \le y\}$ 
denote 
the distribution function
of $Y(w)$, and 
let 
$F_{Y(w)}^{-1}(u) := \inf\{y \in \R :
F_{Y(w)}(y) \ge u\}$ 
denote 
the corresponding 
quantile function for $u \in (0, 1)$.

For probability levels $\tau_{\ell}, \tau_{u} \in [0, 1]$ with
$\tau_{\ell} < \tau_{u}$, 
we define the average quantile treatment effect
(AQTE) over $[\tau_{\ell}, \tau_{u}]$ as
\begin{equation*}
  \theta(\tau_{\ell}, \tau_{u})
  :=
  \frac{1}{\tau_{u} - \tau_{\ell}}
  \bigg\{
  \int_{\tau_{\ell}}^{\tau_{u}} F_{Y(1)}^{-1}(u)\, du
  -
  \int_{\tau_{\ell}}^{\tau_{u}} F_{Y(0)}^{-1}(u)\, du
  \bigg\}.
\end{equation*}
The AQTE measures the difference between the quantile functions of the potential outcomes, averaged over the interval $[\tau_{\ell}, \tau_{u}]$,
and thus characterizes the distributional treatment effect over the specified quantile range.
Moreover, the integral formulation renders the AQTE robust to the choice of selection from the quantile correspondence, since such irregular points form a set of Lebesgue measure zero and do not alter the integral.

The AQTE encompasses well-known treatment effect parameters as
special cases, as established in the following proposition.

\vspace{0.3cm}
\begin{proposition}
  \label{proposition:AQTE}
  Suppose $\E[|Y(w)|] < \infty$. Then
  \begin{enumerate}[label=\textnormal{(\alph*)}, topsep=3pt, itemsep=3pt, align=left, leftmargin=1em, labelsep=0.5em]
    \item $\theta(0, 1) = \E[Y(1) - Y(0)]$;
    \item for any $\tau \in (0, 1)$, 
      the limit point of
      $\theta(\tau_{\ell}, \tau_{u})$ as $(\tau_{\ell}, \tau_{u})
      \to (\tau, \tau)$ with $\tau_{\ell} < \tau_{u}$
      equals $F_{Y(1)}^{-1}(\tau) -
      F_{Y(0)}^{-1}(\tau)$
      if $F_{Y(w)}^{-1}(\cdot)$ is
      continuous at $\tau$ for each $w \in \{0, 1\}$, and
      in general,  
      every limit point of
      $\theta(\tau_{\ell}, \tau_{u} )$ belongs to
      the interval $[F_{Y(1)}^{-1}(\tau) -
      F_{Y(0)}^{-1}(\tau{\scalebox{0.7}{$+$}}),\,
      F_{Y(1)}^{-1}(\tau{\scalebox{0.7}{$+$}}) -
      F_{Y(0)}^{-1}(\tau)]$. 
  \end{enumerate}
\end{proposition}
\vspace{0.3cm}

Proposition \ref{proposition:AQTE} shows that the AQTE nests two canonical treatment effect parameters. 
The AQTE recovers the average
treatment effect by setting $(\tau_{\ell}, \tau_{u}) = (0, 1)$.
It also recovers the $\tau$-th quantile treatment effect in the
limit as $(\tau_{\ell}, \tau_{u}) \to (\tau, \tau)$, provided the quantile functions of both potential outcomes are continuous at $\tau$.

\subsection{Identification}

We consider the identification of the AQTE through the integrated
distribution and quantile functions of the potential outcomes,
defined for each $w \in \{0,1\}$ as
\begin{equation*}
  G_{Y(w)}(y)
  := \int_{-\infty}^{y} F_{Y(w)}(s)\, ds
  \quad 
  \mathrm{and}
  \quad 
  G_{Y(w)}^{\ast}(\tau)
  :=
  \int_{0}^{\tau}
  F_{Y(w)}^{-1}(u)\, du . 
\end{equation*}
By Proposition \ref{proposition:L2-min-1}(b),
the integrated quantile function 
$G_{Y(w)}^{\ast} (\cdot)$ 
is the convex conjugate of  
$G_{Y(w)} (\cdot)$.
The endpoint values are $G_{Y(w)}^{\ast}(0) = 0$ and $G_{Y(w)}^{\ast}(1) = \E[Y(w)]$ for each $w \in \{0,1\}$.
It therefore suffices to identify
$G_{Y(w)}$ for the identification of the AQTE.

We consider identification under random assignment, as in
randomized experiments. Identification under unconfoundedness
conditional on pre-treatment covariates is developed in Appendix B.

\vspace{0.3cm}
\begin{assumption}
  \label{assump:random-assignment}
  The data-generating process satisfies:
  \begin{enumerate}[label=(\alph*), noitemsep, topsep=0pt]
    \item[\textnormal{(a)}] $Y = Y(W)$ almost surely.
    \item[\textnormal{(b)}] $\E[|Y(w)|] < \infty$ for each $w \in
      \{0, 1\}$.
    \item[\textnormal{(c)}] 
    $(Y(0), Y(1)) \independent W$ and
      $0 < \Pr(W = 1) < 1$.
  \end{enumerate}
\end{assumption}
\vspace{0.3cm}

Assumption \ref{assump:random-assignment} collects the standard
conditions of the potential outcomes framework. Condition (a) is
the consistency condition that relates the observed outcome to the
realized potential outcome. Condition (b) ensures that the
integrated distribution function $G_{Y(w)}$ is finite for each
$w \in \{0, 1\}$. Condition (c) combines random assignment of $W$
with overlap.

Under the present setup, the ReLU regression on a constant and $W$ is
saturated and therefore correctly specified, so that
$\E[(y-Y)_{+}| W=w]=G_{Y| W}(y | w)$. Accordingly, in the
absence of misspecification, this section does not distinguish the ReLU projection from the corresponding integrated distribution function. We now state the identification result for the AQTE under
random assignment.

\vspace{0.3cm}
\begin{theorem}
  \label{theorem:AQTE}
  Suppose Assumption \ref{assump:random-assignment} holds. Then,
  $G_{Y(w)}(y) = \E[(y - Y)_{+} | W = w]$
  for each $w \in \{0, 1\}$ and $y \in \R$,
  and 
  for any $\tau_{\ell}, \tau_{u} \in [0, 1]$ with $\tau_{\ell} < \tau_{u}$, 
  \begin{equation*}
    \theta(\tau_{\ell}, \tau_{u})
    = 
    \frac{1}{\tau_{u} - \tau_{\ell}}
    \big \{
      \big ( G_{Y(1)}^{\ast}(\tau_u) - G_{Y(1)}^{\ast}(\tau_\ell)\big) 
      - 
      \big ( G_{Y(0)}^{\ast}(\tau_u) - G_{Y(0)}^{\ast}(\tau_\ell)\big) 
    \big \}. 
  \end{equation*}
 
\end{theorem}
\vspace{0.3cm}

The theorem establishes identification of the AQTE through the
Legendre-Fenchel transform of the integrated distribution
function $G_{Y(w)}(\cdot)$, which is itself identified from the
conditional distribution of $Y$ given $W = w$ under random
assignment, for each $w \in \{0, 1\}$. The AQTE is point-identified
without imposing continuity of the potential outcome distributions 
and captures distributional treatment effects across quantile ranges through the choice of the quantile interval $[\tau_{\ell}, \tau_{u}]$.

\subsection{Estimation and Asymptotic Properties}

Suppose that we observe an i.i.d.\ sample
$\{(W_{i}, Y_{i})\}_{i=1}^{n}$ drawn from the joint distribution of
$(W, Y)$. For each $w \in \{0,1\}$, we estimate the integrated
distribution function by the treatment-group average
\begin{equation*}
  \widehat{G}_{Y(w)}(y)
  :=
  \frac{\sum_{i=1}^{n}\1\{W_i=w\}(y-Y_i)_{+}}
  {\sum_{i=1}^{n}\1\{W_i=w\}}.
\end{equation*}
This is the fitted value from the ReLU regression on a constant and
$W$. Its convex conjugate
$\widehat{G}_{Y(w)}^{\ast}$ is obtained by the Legendre-Fenchel
transformation. For quantile indices $\tau_{\ell}, \tau_{u} \in
[0, 1]$ with $\tau_{\ell} < \tau_{u}$, the AQTE estimator is
\begin{equation}
  \label{eq:aqte-hat}
  \widehat{\theta}(\tau_{\ell}, \tau_{u})
  =
  \frac{1}{\tau_{u} - \tau_{\ell}}
  \big\{
    \big( \widehat{G}_{Y(1)}^{\ast}(\tau_{u})
      - \widehat{G}_{Y(1)}^{\ast}(\tau_{\ell}) \big)
    - \big( \widehat{G}_{Y(0)}^{\ast}(\tau_{u})
      - \widehat{G}_{Y(0)}^{\ast}(\tau_{\ell}) \big)
  \big\}.
\end{equation}

We derive the asymptotic distribution of $\widehat{\theta}(\tau_{\ell},
\tau_{u})$ under the following regularity condition.

\vspace{0.3cm}
\begin{assumption}
  \label{assump:aqte-regularity}
  The following conditions hold:
  \begin{enumerate}[label=(\alph*), noitemsep, topsep=0pt]
    \item[\textnormal{(a)}]
    $\{(W_{i}, Y_{i})\}_{i=1}^{n}$ is i.i.d.\ from the joint
    distribution of $(W, Y)$.
    \item[\textnormal{(b)}]
    $\E[Y(w)^{2}] < \infty$ for each $w \in \{0, 1\}$.
  \end{enumerate}
\end{assumption}
\vspace{0.3cm}

Assumption \ref{assump:aqte-regularity}(a) imposes independent
random sampling. 
Because $W$ is bounded, a square-integrable envelope holds under
the second-moment condition in Assumption
\ref{assump:aqte-regularity}(b). The empirical-process argument of
Theorem \ref{theorem:aym} therefore applies to the treatment-group
averages.
The resulting
weak limit is a zero-mean Gaussian process $\mathbb{B}_{w}(\cdot)$
in $\ell^{\infty}(\Rex)$ for each $w \in \{0,1\}$,
with covariance kernel
\begin{equation*}
\Sigma_{w}(y_{1}, y_{2})
    :=
\frac{1}{\Pr(W = w)}
    \big[
\E[(y_{1} - Y(w))_{+}(y_{2} - Y(w))_{+}]
- 
G_{Y(w)}(y_{1}) 
G_{Y(w)}(y_{2})
    \big],
\end{equation*}
and $(\mathbb{B}_{0}, \mathbb{B}_{1})$ are independent.

\vspace{0.3cm}
\begin{theorem}
  \label{theorem:AQTE-asymp}
  Suppose Assumptions 
  \ref{assump:random-assignment}
  and 
  \ref{assump:aqte-regularity} hold. 
  Then,  for fixed $\tau_{\ell}, \tau_{u} \in [0, 1]$ 
  with $\tau_{\ell} < \tau_{u}$, we have 
  \begin{equation*}
    \sqrt{n}
    \big(
      \widehat{\theta}(\tau_{\ell}, \tau_{u}) 
      -
      \theta(\tau_{\ell}, \tau_{u})\big) \rightsquigarrow
      \dfrac{1}{\tau_{u} - \tau_{\ell}}
      \big\{
        \big( \mathbb{Z}_{1}(\tau_{u}) - \mathbb{Z}_{1}(\tau_{\ell}) \big)
        - \big( \mathbb{Z}_{0}(\tau_{u}) - \mathbb{Z}_{0}(\tau_{\ell}) \big)
      \big\},     
  \end{equation*}
 where
  $
    \mathbb{Z}_{w}(\tau)
    :=
    \lim_{\eta \downarrow 0}
    \sup_{y \in \partial_{\eta} 
    G_{Y(w)}^{\ast}(\tau)}
    \{ -\mathbb{B}_{w}(y) \}
  $
  for each $w \in \{0, 1\}$ and $\tau \in \{\tau_{\ell}, \tau_{u}\}$.
  For $\eta>0$,
  the $\eta$-subdifferential of
  $G_{Y(w)}^{\ast}$ at $\tau\in[0,1]$ is defined by 
  \begin{equation*}
  \partial_{\eta}G_{Y(w)}^{\ast}(\tau)
  :=
  \big \{
  y\in\R:
  G_{Y(w)}^{\ast}(\tau')
  \ge
  G_{Y(w)}^{\ast}(\tau)
  +y(\tau'-\tau)-\eta
  \quad
  \text{for every }\tau'\in[0,1]
  \big \}.
  \end{equation*}
  Also, the delta-method bootstrap 
  based on a consistent estimator of the $\eta$-subdifferential  
  $\partial_{\eta} G_{Y(w)}^{\ast}(\tau)$ for $w \in \{0,1\}$
  consistently estimates the limit law above.

\end{theorem}
\vspace{0.3cm}

Theorem \ref{theorem:AQTE-asymp} characterizes the
limit law as a linear combination of the Hadamard directional
derivatives of the Legendre-Fenchel transform evaluated at
$G_{Y(0)}$ and $G_{Y(1)}$. 
At points $\tau$ where $F_{Y(w)}^{-1}$
is continuous, the subdifferential
$\partial G_{Y(w)}^{\ast}(\tau)$ collapses to a singleton and
$\mathbb{Z}_{w}(\tau) = -\mathbb{B}_{w}(F_{Y(w)}^{-1}(\tau))$, and the limit
law is a centered Gaussian. 
At points where $F_{Y(w)}^{-1}$ is discontinuous at $\tau$,
$\partial G_{Y(w)}^{\ast}(\tau)$ has positive length, and
$\mathbb{Z}_{w}(\tau)$ is the supremum of a Gaussian process over that
subdifferential.\footnote{
  Notice that since the ReLU regression is saturated, 
  $G_{Y(w)}(\cdot)$ is
  convex and thus 
  the $\eta$-subdifferential 
  $\partial_{\eta} G_{Y(w)}^{\ast}(\tau)$ 
  coincides with the unconditional counterpart of the approximate-maximizer sets
  $S_{\eta}(\tau, x)$ of Section~\ref{sec:estimation}.
}
At the endpoints, 
$\mathbb{Z}_{w}(0) = 0$ and 
$\mathbb{Z}_{w}(1) = -\mathbb{B}_{w}(+\infty)$,
so for $(\tau_{\ell}, \tau_{u}) = (0, 1)$ the limit reduces to the
standard difference-in-means asymptotics for the average treatment
effect.
To ensure valid bootstrap inference, the delta-method bootstrap is required.
Following Section \ref{sec:estimation}, the delta-method bootstrap based on the sample
$\eta_n$-subdifferential
$\partial_{\eta_n}\widehat{G}_{Y(w)}^{\ast}(\tau)$ consistently
estimates the limit law above 
with $\eta_{n} \to 0$ and $\sqrt{n}\, \eta_{n} \to \infty$.

\subsection{AQTE Estimation with Covariates}
\label{sec:aqte-covariates}

The preceding subsection considers the specification without control
variables. With a slight abuse of notation, we now use $X$ to denote a
vector of pre-treatment covariates and suppose that
$(Y(0),Y(1))\independent W \mid X$ with overlap.
Let $b(X)$ be a finite-dimensional vector of basis functions of $X$
that includes a constant.

For each $y\in\R$, we consider the fully interacted ReLU regression
\begin{equation*}
(y-Y)_{+}
=
b(X)^{\top}\gamma_{0}(y)
+
Wb(X)^{\top}\delta_{0}(y)
+
\epsilon(y),
\end{equation*}
where 
$\gamma_{0}(y)$
and 
$\delta_{0}(y)$
are unknown coefficient vectors 
and 
$\epsilon(y)$
is an error term. 
Under the correct specification, 
$\E[\epsilon(y) | W,X]=0$.
This conditional mean
restriction is satisfied when $X$ has finite support and $b(X)$ is
chosen as a saturated set of stratum indicators.
Moreover, under the selection-on-observables assumption, 
$
  \E[(y-Y)_{+}| W=w,X=x]
  =
  G_{Y(w)| X}(y | x),
$
for each $w\in\{0,1\}$,
where 
$G_{Y(w)|X}(y|x)$
denotes the integrated conditional distribution function 
of $Y(w)$ given $X$.

Let
$\hat{\gamma}(y)$
and 
$\hat{\delta}(y)$ denote the least-squares estimators. 
For each treatment status $w \in \{0,1\}$, 
we estimate $G_{Y(w) | X}(y| x)$ 
by fitted value $\widehat{G}_{Y(w)|X}(y|x)$ from the regression above.
Then, the marginal integrated distribution function is recovered by
integrating the conditional function over the distribution of $X$.
Accordingly, we use 
\begin{equation*}
\widehat{G}_{Y(w)}(y)
=
\frac{1}{n}\sum_{i=1}^{n}
\widehat{G}_{Y(w)| X}(y | X_i),
\end{equation*}
to construct its conjugate 
$\widehat{G}_{Y(w)}^{\ast}(\cdot)$
and the marginal AQTE estimator as defined above.
This fully interacted specification is used for the empirical application in Section \ref{sec:application}.

In Appendix B, we provide the assumptions and identification results
for the conditional AQTE, together with detailed asymptotic results.
In particular, Theorem \ref{theorem:marginalized-asymp} establishes the
weak convergence of the marginalized estimator and the validity of the
corresponding delta-method bootstrap.

\section{Application}
\label{sec:application}

In this empirical application, we examine the effect of winning the Medicaid lottery on health-care utilization. The 2008 expansion
of the Oregon Medicaid program offered enrollment to a randomly
selected subset of low-income uninsured adults from a waitlist. The resulting Oregon Health Insurance Experiment is a
randomized evaluation of the distributional effects of this lottery offer \citep{finkelstein2012oregon}. The original analysis of
\citet{finkelstein2012oregon} reports intention-to-treat estimates
from a linear regression of the outcome on the lottery offer.
\citet{chernozhukov2020generic} revisit the experiment and report
quantile treatment effects for the count outcome through
quantile bands.

We consider $W = 1$ for individuals whose household was selected in the lottery and $W = 0$ otherwise.
The outcome $Y$ is the number of outpatient visits in the
six-month period preceding the survey. 
The lottery offer was randomly assigned, 
conditional on listed household size.
We collect the design features in
the vector $X$, which records six survey-wave dummies, two
household-size dummies, and their interactions,
and thus 
the regression model is fully saturated.
Conditional on
$X$, the assignment is independent of the potential outcomes,
$(Y(0), Y(1)) \independent W \mid X$. The marginal distribution of
$Y(w)$ is therefore identified by averaging the conditional
distribution of $Y$ given $(W = w, X)$ over the marginal
distribution of $X$. The estimand is the average quantile
treatment effect $\theta(\tau_{\ell}, \tau_{u})$ of
Section \ref{sec:treatment}. It has the intention-to-treat
interpretation as the design-averaged quantile difference between
the offered and not-offered populations on the subinterval
$[\tau_{\ell}, \tau_{u}]$.

Because $Y$ is a count variable, its distribution function is a step function, and the quantile function inherits the jumps and flat regions.
For inference in this setting, \citet{chernozhukov2020generic} 
invert distribution-regression bands to construct confidence bands for the quantile functions and then obtain confidence bands for quantile effects via the Minkowski difference of these bands.
We instead estimate the integrated distribution function for each treatment group by ordinary least squares, and then integrate out the design controls through their empirical distribution. 
The integrated quantile function is recovered as its convex conjugate. This yields the average quantile treatment effect as a difference of conjugates evaluated at two probability levels.
Inference uses the delta-method bootstrap of Section \ref{sec:estimation} based on the empirical subdifferential interval and requires no
numerical differentiation.

Our data come from the public release of the Oregon Health
Insurance Experiment \citep{finkelstein2012oregon}. 
The sample
consists of 23,441 individuals across 20,747 households on the lottery list, comprising 11,651 individuals with $W=1$ and 11,790 individuals with $W=0$.
The outcome $Y$ is right-skewed and concentrated on small
integers. It has a substantial mass at zero and a sparse upper
tail extending to several dozen visits. The empirical
distributions of $Y$ for the treatment and control groups are reported in
Figure \ref{fig:application-histogram}.

We implement the fully interacted ReLU regression and marginalization
procedure described in Section \ref{sec:aqte-covariates}. For each
$y \in \R$, we regress $(y-Y_i)_{+}$ on
$X_i$ and $W_iX_i$, where $X_i$ includes a constant. 
We then average the fitted conditional
integrated distribution functions over the empirical distribution of
$X$.
Finally, we apply the Legendre-Fenchel transform to obtain
$\widehat{G}_{Y(w)}^{\ast}$.

We report the average quantile treatment effect
$\widehat{\theta}(\tau, \tau + 0.10)$ on the grid
$\tau \in \{0, 0.1, \ldots, 0.9\}$. Inference uses the
delta-method bootstrap described in
Section \ref{sec:estimation}. We draw the bootstrap weights at
the household level following Appendix C. Sampling weights from
the public release enter the ReLU regression as observation weights.

\begin{figure}[htbp]
  \centering
  \caption{Empirical distribution of outpatient visits by treatment status.}
  \includegraphics[width=0.9\textwidth]{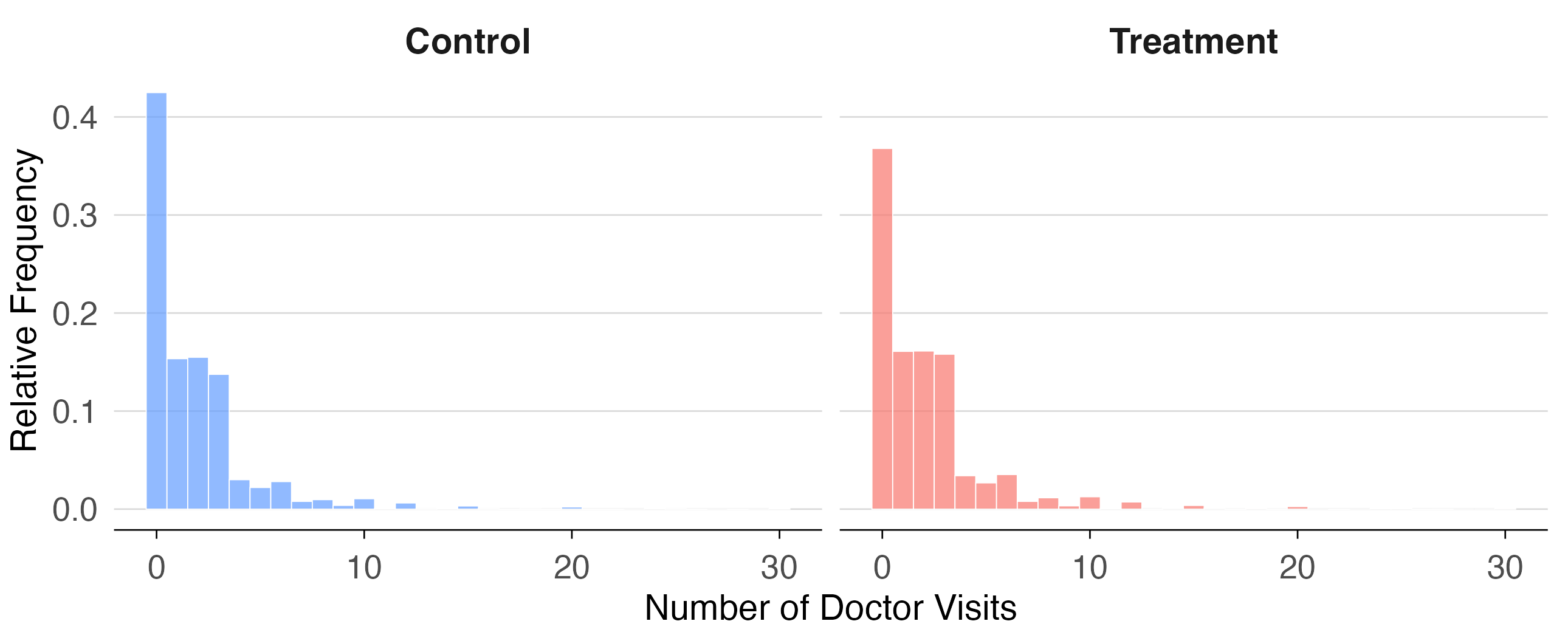}
  \label{fig:application-histogram}
  \begin{minipage}{0.92\textwidth}
    \small
    \begin{spacing}{1.0}
    \textit{Notes:} The figure reports the empirical histogram of
    the number of outpatient visits in the six months preceding
    the survey. The left panel pools the lottery losers ($W = 0$)
    and the right panel the lottery winners ($W = 1$). The
    horizontal axis is truncated at thirty visits.
    \end{spacing}
  \end{minipage}
\end{figure}

Figure \ref{fig:application-histogram} shows that the outcome
distribution is concentrated on a small number of integer values
and shares the same support across the two groups. The
point mass at zero is visible in both panels and is smaller in
the lottery-winner panel than in the lottery-loser panel. The
mass at zero and the sparse upper tail motivate working with the
integrated distribution function rather than with the quantile
function directly.

\begin{figure}[H]
  \centering
  \caption{Estimated integrated quantile functions.}
  \includegraphics[width=0.45\textwidth]{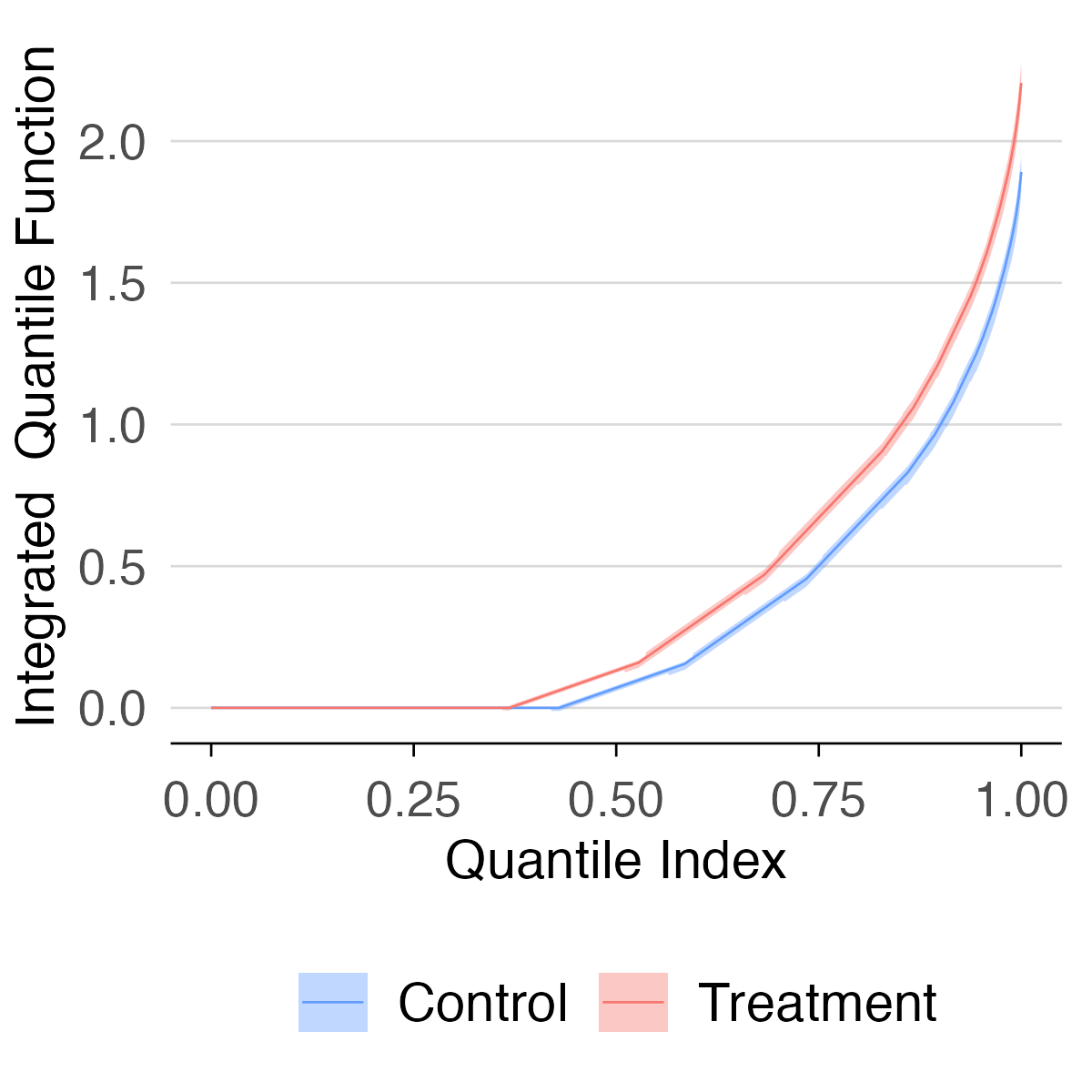}
  \label{fig:application-conjugate}
  \begin{minipage}{0.9\textwidth}
    \small
    \begin{spacing}{1.0}
    \textit{Notes:} 
    The figure reports the estimated integrated
    quantile function  for
    $w \in \{0, 1\}$ on $\tau \in (0, 1)$. The shaded bands
    are pointwise $95\%$ confidence intervals based on the
    delta-method bootstrap with 
    bootstrap weights drawn at the household level.
    \end{spacing}
  \end{minipage}
\end{figure}

Figure \ref{fig:application-conjugate} reports the estimated
integrated quantile functions $\widehat{G}_{Y(w)}^{\ast}(\tau)$ for the
two treatment groups. Both curves are nondecreasing and convex by
construction. At points of differentiability, the slope at a
probability level $\tau$ identifies the marginal quantile
$F_{Y(w)}^{-1}(\tau)$. Both curves are flat over the lower
portion of the unit interval, reflecting the point mass at zero
in the marginal distribution of $Y(w)$, and start to rise around
$\tau \approx 0.35$. The treatment-group curve sits uniformly above the control-group curve over the upper portion. 
The change in vertical separation between probability levels
$\tau_{\ell}$ and $\tau_{u}$, divided by $\tau_{u} - \tau_{\ell}$,
equals the average quantile treatment effect over
$[\tau_{\ell}, \tau_{u}]$.

Figure \ref{fig:application-aqte} reports the estimated average
quantile treatment effect across the probability range. The
estimates are near zero on the three lowest subintervals,
mirroring the flat region of Figure
\ref{fig:application-conjugate} below $\tau \approx 0.35$. They are positive on every subinterval above $\tau = 0.3$ and are
larger at higher quantiles on average. The lottery offer 
therefore affects the upper tail of utilization more strongly than
the lower tail. The confidence intervals widens with $\tau$ and is widest on the topmost subinterval, reflecting the sparse upper tail of the outcome distribution.

\begin{figure}[H]
  \centering
  \caption{Average quantile treatment effect.}
  \includegraphics[width=0.7\textwidth]{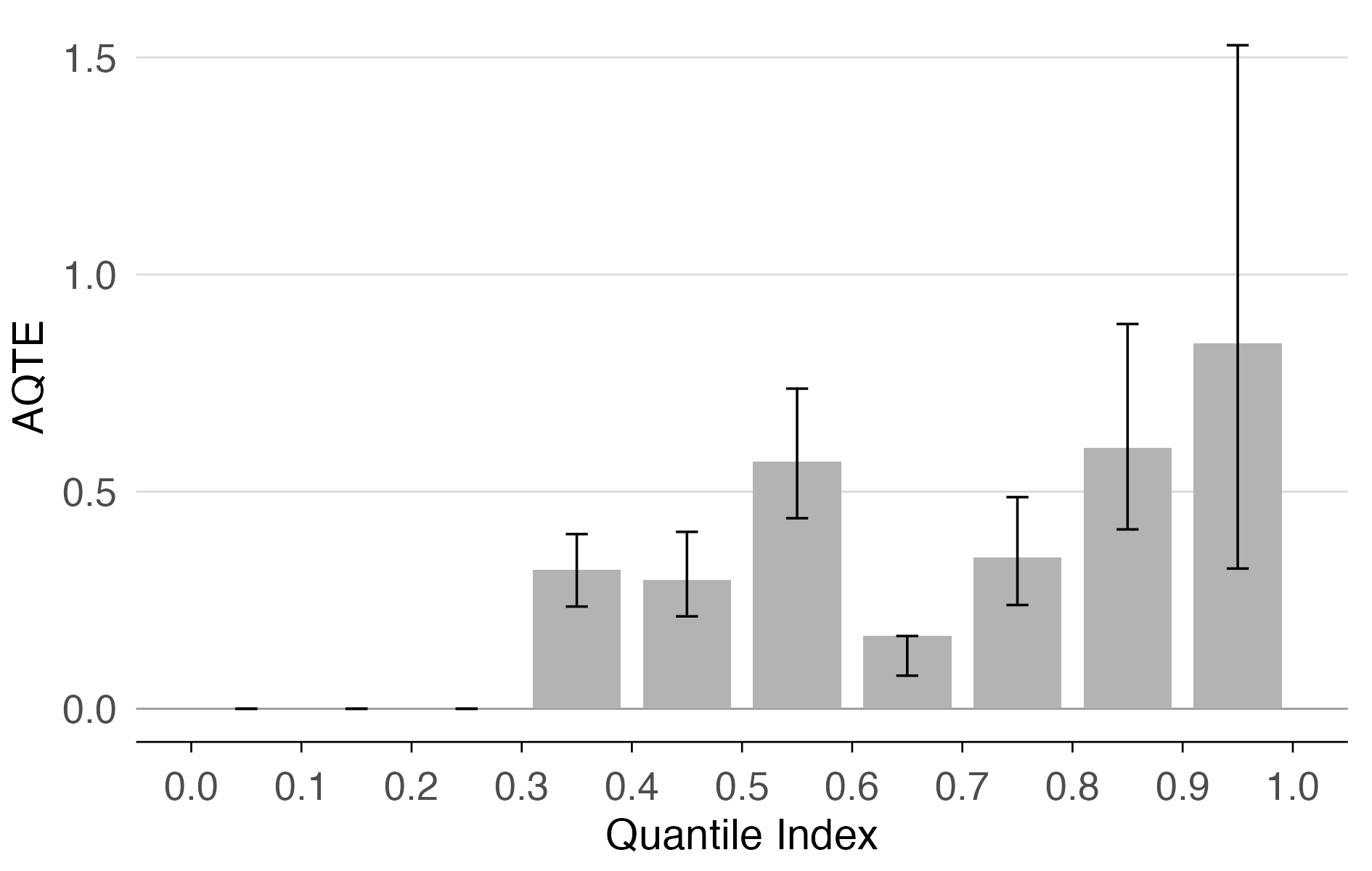}
  \label{fig:application-aqte}
  \begin{minipage}{0.9\textwidth}
    \small
    \begin{spacing}{1.0}
    \textit{Notes:} The figure reports the estimated average
    quantile treatment effect
    $\{\widehat{\theta}(\tau, \tau + 0.10) : \tau \in
    \{0.0, 0.1, \ldots, 0.9\}\}$ together with the pointwise
    $95\%$ confidence intervals based on the delta-method bootstrap with bootstrap weights drawn at the
    household level. Each bar is plotted at the midpoint of the
    corresponding subinterval.
    \end{spacing}
  \end{minipage}
\end{figure}

The estimated pattern has a clear distributional interpretation. 
The flat lower portion of both estimated integrated quantile functions reflects the large share of individuals with zero outpatient visits across the treatment and control groups over the six-month window.
Thus, the AQTE estimates are close to zero over the three lowest quantile subintervals.
Over higher subintervals, the estimated AQTEs are positive and tend to be
 larger toward the upper tail.
These range-specific distributional contrasts document distributional effects of the lottery offer across quantile ranges.

\section{Conclusion}
\label{sec:conclusion}

This paper develops a regression model with the rectified linear
unit transformation applied to the outcome variable. 
The ReLU regression estimates the best linear approximation to the integrated conditional distribution function in closed form.
Through the Legendre-Fenchel transform, the corresponding convex conjugate approximates the integrated conditional quantile function. 
Both approximations are exact under correct specification, 
which holds for saturated specifications.
The framework requires only finite moments and standard rank conditions,
and it accommodates outcome distributions that are discrete, mixed,
or otherwise non-continuous without modification.

We establish the uniform asymptotic distribution of the estimator
as a process indexed by the outcome location, and we develop
inference for the conjugate functional through the delta method
for Hadamard directionally differentiable maps. As an application,
the framework identifies average quantile treatment effects over
arbitrary quantile subintervals under random assignment of a
binary treatment, with bootstrap inference based on the same
delta-method procedure.

\clearpage 
\setstretch{0.1}
\bibliographystyle{chicago}
\bibliography{ReLU}

\clearpage

\setstretch{1.2}
\setcounter{section}{0} 
\setcounter{figure}{0} 
\setcounter{table}{0} 
\setcounter{equation}{0} 
\setcounter{lemma}{0}\setcounter{page}{1}
\setcounter{theorem}{0}\setcounter{page}{1}
\setcounter{proposition}{0}\setcounter{page}{1}

\renewcommand{\thepage}{A-\arabic{page}}

\renewcommand{\theequation}{A.\arabic{equation}}
\renewcommand{\theHequation}{A.\arabic{equation}}
\renewcommand{\thelemma}{A.\arabic{lemma}}
\renewcommand{\theHlemma}{A.\arabic{lemma}}
\renewcommand{\thetheorem}{A.\arabic{theorem}}
\renewcommand{\theHtheorem}{A.\arabic{theorem}}
\renewcommand{\theproposition}{A.\arabic{proposition}}
\renewcommand{\theHproposition}{A.\arabic{proposition}}
\renewcommand{\thedefinition}{A.\arabic{definition}}
\renewcommand{\theHdefinition}{A.\arabic{definition}}

\renewcommand\thesection{\Alph{section}} 
\renewcommand\theHsection{A}
\renewcommand\thesubsection{A.\arabic{subsection}}
\renewcommand\theHsubsection{A.\arabic{subsection}}
\renewcommand\thefigure{A.\arabic{figure}}
\renewcommand\theHfigure{A.\arabic{figure}}
\renewcommand\thetable{A.\arabic{table}}
\renewcommand\theHtable{A.\arabic{table}}

\begin{center}
  \section*{Appendix}
\end{center}

\section*{Appendix A: Theoretical Results}

In what follows, let $\lesssim$ denote inequality up to a constant
independent of $n$.
For a symmetric matrix $A$, $\lambda_{\min}(A)$ denotes its smallest
eigenvalue.
Let $e_{k}$ be a vector with its $k$th entry being 1 and the other entries being 0.

We equip the extended real line $\Rex$ with the metric
$\rho: \Rex^2 \to [0, \pi]$, defined by, 
for $y, y' \in \Rex$, 
\begin{equation*}
  \rho(y,y')
  :=
  |\arctan(y)-\arctan(y')|,    
\end{equation*}
with the convention  
$\arctan(\pm \infty) = \pm \pi /2$. 
The metric space 
$(\Rex, \rho)$ 
is shown to be compact 
and 
allows us to consider asymptotic properties of the process on the whole extended real line 
while $\rho$ induces the usual topology on $\R$
\citep[see e.g., p.~26,][]{van1996weak}.

\subsection{Technical Lemmas}

The technical lemmas are available upon request.

\subsection{Proofs of the Main Results}

\vspace{0.5cm}
\begin{proof}[\normalfont \textbf{Proof of Proposition \ref{proposition:L2-min-1}}]

  \textbf{Part (a).}
  Fix $x \in \X$ and $ y \in \R$. 
  Since 
  $\E[ |Y| ] < \infty$
  and thus 
  $\E[|Y| | X = x] < \infty$ almost surely,
  we have that 
  $G_{Y | X}(y|x) \equiv \int_{-\infty}^{y} F_{Y | X}(s | x) \, ds$ is finite. 
  For any $s \le y$, we can show that  
  $y - s = \int_{-\infty}^{y}
  \1\{s \le v\} \, dv$,  and applying Fubini's
  theorem, we have 
  \begin{align*}
    \E[(y - Y)_{+} | X = x]
    &=
    \int_{-\infty}^{y} (y - s) \, dF_{Y | X}(s | x)\\
    &=
    \int_{-\infty}^{y}
      \bigg( \int_{-\infty}^{y} \1\{s \le v\} \, dv \bigg)
      dF_{Y | X}(s | x)\\
    &=
    \int_{-\infty}^{y} F_{Y | X}(v | x) \, dv . 
  \end{align*}
  Here, the integrand $F_{Y | X}(\cdot | x)$ is nondecreasing,
  and thus  
  $G_{Y | X}(\cdot | x)$ is convex. 
  The right derivative of
  $G_{Y | X}(\cdot | x)$ at $y$ equals $F_{Y | X}(y | x)$ by the right
  continuity of the conditional distribution function, 
  whereas its left
  derivative equals $F_{Y | X}(y{\scalebox{0.7}{$-$}} | x)$. For a convex function
  on $\R$, the subdifferential at any point coincides with the closed
  interval bounded by its one-sided derivatives \citep[Theorem
  24.1]{rockafellar1997convex}. Thus, 
  $\partial G_{Y | X}(y | x) =
  [F_{Y | X}(y{\scalebox{0.7}{$-$}} | x),\, F_{Y | X}(y | x)]$.

  \textbf{Part (b).}
  The conjugate function $G_{Y | X}^{\ast}(\cdot | x)$ is a pointwise supremum of
  affine functions and thus convex. 
  The map $G_{Y | X}(\cdot | x)$ is
  proper closed convex with $\lim_{y \to -\infty} G_{Y | X}(y | x) = 0$,
  so by Theorem 23.5 of \citet{rockafellar1997convex} the
  subdifferential inversion
  \begin{equation*}
    \tau \in \partial G_{Y | X}(y | x)
    \iff
    y \in \partial G_{Y | X}^{\ast}(\tau | x)
  \end{equation*}
  holds for every $y \in \R$ and $\tau \in (0,1)$. Combining this with
  Part (a) shows that $y \in \partial G_{Y | X}^{\ast}(\tau | x)$ if and
  only if $\tau \in [F_{Y | X}(y{\scalebox{0.7}{$-$}} | x), F_{Y | X}(y | x)]$. The set
  of such $y$ is the closed interval
  $[F_{Y | X}^{-1}(\tau | x),\, 
  F_{Y | X}^{-1}(\tau \scalebox{0.7}{$+$} | x)]$
  by the definition of the generalized inverse, hence
  \begin{equation*}
    \partial G_{Y | X}^{\ast}(\tau | x)
    =
    [F_{Y | X}^{-1}(\tau | x),\,
    F_{Y | X}^{-1}(\tau {\scalebox{0.7}{$+$}} | x)].
  \end{equation*}
  Since $G_{Y | X}^{\ast}(\cdot | x)$ is convex and its subdifferential
  contains the nondecreasing function $F_{Y | X}^{-1}(\cdot | x)$, the
  function $G_{Y | X}^{\ast}(\cdot | x)$ is the primitive of
  $F_{Y | X}^{-1}(\cdot | x)$ on $(0, 1)$. Evaluating
  $G_{Y | X}^{\ast}(0 | x) = \sup_{y \in \R}\{-G_{Y | X}(y | x)\} = 0$, since
  $G_{Y | X}(\cdot | x) \ge 0$ and $\lim_{y \to -\infty} G_{Y | X}(y | x) =
  0$. 
  At $\tau=1$, 
  using the identity
  $y-(y-Y)_{+}=\min\{y,Y\}$, we can write 
  $
  G_{Y|X}^{\ast}(1|x)
  =
  \sup_{y\in\R}
  \int_{0}^{1}
  \min\{y,F_{Y|X}^{-1}(u|x)\} du.
  $
  The integral 
  $  \int_{0}^{1}
  \min\{y,F_{Y|X}^{-1}(u|x)\} du
  $
  is nondecreasing in $y$, and its integrand
  is dominated in absolute value by
  $|F_{Y|X}^{-1}(u|x)|$
  for $y\geq0$.
  Dominated convergence therefore yields
  $
  G_{Y|X}^{\ast}(1|x)
  =
  \int_{0}^{1}F_{Y|X}^{-1}(u|x)du
  $.
  Therefore,
  \begin{equation*}
    G_{Y | X}^{\ast}(\tau | x)
    =
    \int_{0}^{\tau} F_{Y | X}^{-1}(u | x) \, du,
  \end{equation*}
  for all $\tau \in [0, 1]$.
\end{proof}
\vspace{0.5cm}

Recall the covariance matrix 
$\Sigma(y_{1}, y_{2}) := \mathrm{Cov} \big( X \epsilon(y_{1}),\, X \epsilon(y_{2}) \big)$
as defined in Theorem \ref{theorem:aym}. 
We extend the error to the endpoints of $\Rex$ by
considering 
$\epsilon(-\infty) := 0$
and
$\epsilon(+\infty) := -\big( Y - X^{\top} Q_{X}^{-1} \E[XY] \big)$,
which are the $L^{2}(P)$ limits of $\epsilon(y)$ as $y \to -\infty$ and as
$y \to +\infty$, respectively.

\vspace{0.5cm}
\begin{proof}[\textnormal{\textbf{Proof of Theorem \ref{theorem:aym}}}]

  Under Assumptions \ref{assump:dgp} and \ref{assump:regularity},
  $n^{-1} \sum_{i=1}^{n} X_{i} X_{i}^{\top}$ is invertible with
  probability approaching one. On this event, for every $y \in \R$,
  \begin{equation*}
    \sqrt{n}\big( \hat{\beta}(y) - \beta_{0}(y) \big)
    =
    \bigg ( 
      \frac{1}{n}
      \sum_{i=1}^{n} X_{i} X_{i}^{\top}
    \bigg )^{-1}
    \frac{1}{\sqrt{n}}
    \sum_{i=1}^{n}
    f_{y}(X_{i}, Y_{i}), 
  \end{equation*}
  where 
  $ f_{y} \in \mathcal{F}$
  with 
  $ \mathcal{F}
  $ 
  denoting the class of $\R^p$-valued functions defined in
  Lemma A2.
  Since $\E[f_{y}(X, Y)] = 0$ for every $y \in \Rex$, we have
  $\E \big\| \mathbb{G}_{n} f_{y} - \mathbb{G}_{n} f_{\pm\infty} \big\|^{2}
  \le \| f_{y} - f_{\pm\infty} \|_{P,2}^{2}$,
  which vanishes as $y \to \pm\infty$ by Lemma
  A2(b). 
  On the event above, the inverse sample
  second-moment matrix is finite and does not depend on $y$.
  Thus, the right-hand side converges in probability to its endpoint
  value as $y\to\pm\infty$, extending the preceding representation to
  every $y\in\Rex$.

  By Lemma A3,
  \begin{equation*}
    \{\mathbb G_n f_y:y\in\Rex\}
    \rightsquigarrow
    \{\mathbb G f_y:y\in\Rex\}
    \quad\text{in }\ell^\infty(\Rex)^p,
  \end{equation*}
  where $\{\mathbb G f_y:y\in\Rex\}$ is a tight mean-zero Gaussian
  process with $\rho$-uniformly continuous sample paths.
   Since $\E[f_y(X,Y)]=0$ for every $y\in\Rex$, the covariance function
  of the limiting process is, for every 
  $y_{1}, y_{2} \in \Rex$,
  \begin{equation*}
    \E[f_{y_1}(X,Y)f_{y_2}(X,Y)^\top]
    =
    \E[XX^\top\epsilon(y_1)\epsilon(y_2)]
    =
    \Sigma(y_1,y_2).
  \end{equation*}
  Moreover, Lemma A2(a) implies that
  $\sup_{y_1,y_2\in\Rex}\|\Sigma(y_1,y_2)\|
  \le \E[F(X,Y)^2]<\infty$.

  Under Assumptions \ref{assump:dgp}(a) and \ref{assump:regularity}(a),
  the law of large numbers shows that 
  $n^{-1}\sum_{i=1}^{n} X_{i} X_{i}^{\top} = Q_{X} + o_{p}(1)$ with
  $Q_{X} := \E[X X^{\top}]$.
  By the weak
  convergence established above,
  $\sup_{y \in \Rex}\|\mathbb{G}_n f_y\| = O_p(1)$.  
  Thus, 
  an application of continuous mapping theorem and Slutsky's theorem yields 
  \begin{equation*}
    \sqrt{n}\big(\hat\beta(\cdot) - \beta_0(\cdot)\big)
    \rightsquigarrow
    \mathbb{B}(\cdot)
  \quad\text{in } \ell^\infty(\Rex)^{p}, 
  \end{equation*}
  where 
  $\mathbb{B}$ is the tight mean-zero $\R^p$-valued Gaussian process   
  defined by $\mathbb{B}(y) := Q_X^{-1}\,\mathbb{G}f_y$ 
  with
  $f_y \in \mathcal{F}$.

  For the second conclusion, fix $x \in \X$. For any $y \in \Rex$,
  $\widehat{G}_{Y | X}(y | x) - \widetilde{G}_{Y | X}(y | x) = x^{\top}\big(
  \hat{\beta}(y) - \beta_{0}(y) \big)$. The continuous mapping theorem
  applied to $\sqrt{n}\big(\hat{\beta}(\cdot) - \beta_{0}(\cdot)\big)
  \rightsquigarrow \mathbb{B}(\cdot)$ yields
  $\sqrt{n}\big( \widehat{G}_{Y | X}(\cdot | x) 
  - 
  \widetilde{G}_{Y | X}(\cdot | x) \big)
  \rightsquigarrow x^{\top}\mathbb{B}(\cdot)$ in
  $\ell^{\infty}(\Rex)$.
\end{proof}
\vspace{0.5cm}

\vspace{0.5cm}
\begin{proof}[\textnormal{\textbf{Proof of Theorem
\ref{theorem:L-weak}}}]

  \textbf{(a)}
  Fix $x \in \X$ and $\tau \in [0,1]$. 
  The limit calculations in the proof of
  Lemma A2(b)
  with 
  the continuity of
  $\widetilde{G}_{Y|X}(\cdot|x)$ imply that
  $y\mapsto\tau y-\widetilde{G}_{Y|X}(y|x)$ is bounded from above on $\R$,
  and thus $\widetilde{G}_{Y|X}^{*}(\tau|x)$ is finite.
  Analogously, 
  $\widehat{G}_{Y|X}^{*}(\tau|x)$ is finite
  on the event that the sample second-moment matrix is invertible.
  Thus, both conjugates are well-defined.
  By the definition of the Legendre-Fenchel transform, we can write 
  \begin{equation*}
      \widehat{G}_{Y | X}^{\ast}(\tau | x)
      - 
      \widetilde{G}_{Y | X}^{\ast}(\tau | x)
      =
      \sup_{y \in \R}
      \big\{
        \tau y 
        - 
        \widehat{G}_{Y | X}(y | x) 
        - 
        \widetilde{G}_{Y | X}^{\ast}(\tau | x)
      \big\} . 
  \end{equation*} 
  Adding and subtracting $\widetilde{G}_{Y | X}(y | x)$, we have 
  \begin{equation*}
          \tau y 
        - 
        \widehat{G}_{Y | X}(y | x) 
        - 
        \widetilde{G}_{Y | X}^{\ast}(\tau | x)
        = 
      -  \big (  
        \widehat{G}_{Y | X}(y | x) 
        - 
        \widetilde{G}_{Y | X}(y | x)
        \big ) 
        - 
        \Delta(y, \tau),      
  \end{equation*}
  where 
  $\Delta(y, \tau):=
        \widetilde{G}_{Y | X}(y | x) 
        +
        \widetilde{G}_{Y | X}^{\ast}(\tau | x)
        - \tau y 
  $
  is called the Fenchel-Young gap and is nonnegative for every $y \in \R$.
  It follows that
  \begin{equation*}
      \widehat{G}_{Y | X}^{\ast}(\tau | x)
      - 
      \widetilde{G}_{Y | X}^{\ast}(\tau | x)
      =
      \sup_{y \in \R}
      \sup_{ \eta >0: \eta \ge \Delta(y, \tau)}
      \big\{
      -  \big (  
        \widehat{G}_{Y | X}(y | x) 
        - 
        \widetilde{G}_{Y | X}(y | x)
        \big ) 
        - 
        \eta 
      \big \} . 
  \end{equation*} 
  Here, the two iterated suprema range over the same admissible pairs
  $(y, \eta)$, so their order can be exchanged. 
  \begin{equation*}
  \eta 
  \ge 
  \Delta(y,\tau)
  \quad\Longleftrightarrow\quad
  \tau y-\widetilde{G}_{Y| X}(y|x)
  \ge
  \widetilde{G}_{Y| X}^{\ast}(\tau|x)-\eta.
  \end{equation*}
  Thus, multiplying both sides by $\sqrt{n}$, we obtain 
  \begin{equation*}
    \sqrt{n}
    \big (  
      \widehat{G}_{Y | X}^{\ast}(\tau | x)
      - 
      \widetilde{G}_{Y | X}^{\ast}(\tau | x)
      \big)
      =
      \sup_{\eta > 0}
      \sup_{y \in S_{\eta}(\tau, x)}
      \big\{
      -  
      \sqrt{n}
      \big (  
        \widehat{G}_{Y | X}(y | x) 
        - 
        \widetilde{G}_{Y | X}(y | x)
        \big ) 
        - 
        \sqrt{n} \, \eta 
      \big \} . 
  \end{equation*} 
  Let 
  $s_{n} := \sqrt{n}$ 
  and 
  $h_{n} 
  := 
  s_{n}
  \big(
    \widehat{G}_{Y | X}(\cdot | x)
    -
    \widetilde{G}_{Y | X}(\cdot | x)
  \big)$.
  We rewrite the equation above, 
  using the notation of Lemma A4:
  \begin{equation*}
    \sqrt{n}
    \big(
      \widehat{G}_{Y | X}^{\ast}(\tau | x)
      -
      \widetilde{G}_{Y | X}^{\ast}(\tau | x)
    \big)
    =
    \sup_{\eta > 0}
    \bigl\{
      \Psi_{\eta}(h_{n}) - s_{n}\,\eta
    \bigr\}.
  \end{equation*}
    By Theorem \ref{theorem:aym},
    $
      \sqrt{n}
      \big(
        \widehat{G}_{Y| X}(\cdot|x)
        -
        \widetilde{G}_{Y| X}(\cdot|x)
      \big)
      \rightsquigarrow
      x^{\top}\mathbb{B}(\cdot)
    $
    in $\ell^{\infty}(\Rex)$,
    where $x^{\top}\mathbb{B}(\cdot)$ has sample paths in $C(\Rex)$.
  Lemma A4
  together with 
  the extended continuous-mapping theorem
  \citep[Theorem~1.11.1]{van1996weak} 
  yields  
  \begin{equation*}
    \sup_{\eta > 0}
    \bigl\{
      \Psi_{\eta}
      (
      h_{n}
      )
      - \sqrt{n}\, \eta
    \bigr\}
    \rightsquigarrow
    \lim_{\eta \downarrow 0}
    \Psi_{\eta}
    ( x^{\top} \mathbb{B} ).    
  \end{equation*}
  Since 
  $
    \lim_{\eta \downarrow 0}
    \Psi_{\eta}
    ( x^{\top} \mathbb{B} )
    = \mathbb{Z}(\tau|x)   
  $,
  we obtain 
  \begin{equation*}
    \sqrt{n}
    \big(
      \widehat{G}_{Y| X}^{\ast}(\tau|x)
      -
      \widetilde{G}_{Y| X}^{\ast}(\tau|x)
    \big)
    \rightsquigarrow
    \mathbb{Z}(\tau|x).
  \end{equation*}

  Because $J$ is finite, the same argument applies jointly to
  $\tau_1,\ldots,\tau_J$, and hence
  \begin{equation*}
    \big\{
      \sqrt{n}
      \big(
        \widehat{G}_{Y| X}^{\ast}(\tau_j|x)
        -
        \widetilde{G}_{Y| X}^{\ast}(\tau_j|x)
      \big)
    \big\}_{j=1}^{J}
    \rightsquigarrow
    \big\{
      \mathbb{Z}(\tau_j|x)
    \big\}_{j=1}^{J}.
  \end{equation*}
  Therefore, we obtain the desired result.

  \textbf{(b)}
  Fix a closed set $T \subseteq [0,1]$
  on which the stated differentiability condition holds. 
  Let $h\in C(\Rex)$. 
  For every $\tau\in T$ and $\eta>0$,
  \begin{equation*}
    \Big|
      \sup_{y\in S_{\eta}(\tau,x)}\{-h(y)\}
      +h\big(\nabla\widetilde{G}_{Y| X}^{\ast}(\tau| x)\big)
    \Big|
    \le
    \sup_{y\in S_{\eta}(\tau,x)}
    \big|
      h(y)-h\big(\nabla\widetilde{G}_{Y| X}^{\ast}(\tau| x)\big)
    \big| .
  \end{equation*}
  Since $(\Rex,\rho)$ is compact, $h$ is uniformly continuous with respect
  to $\rho$, and the right-hand side thus vanishes uniformly in
  $\tau\in T$ by the assumed condition. 
Thus, the argument in part (a) applies uniformly over
$\tau\in T$, and the Legendre-Fenchel transform, viewed as a map
taking values in $\ell^\infty(T)$, is Hadamard directionally
differentiable tangentially to $C(\Rex)$, with derivative
$
h \mapsto 
-h\bigl(
\nabla\widetilde{G}_{Y| X}^{\ast}(\cdot| x)
\bigr).
$
Since this derivative is linear and continuous, the transformation is
Hadamard differentiable. Theorem \ref{theorem:aym} and the functional
delta method then yield
\begin{equation*}
\sqrt{n}\bigl(
\widehat{G}_{Y| X}^{\ast}(\cdot| x)
-
\widetilde{G}_{Y| X}^{\ast}(\cdot| x)
\bigr)
\rightsquigarrow
-x^{\top}\mathbb{B}\bigl(
\nabla\widetilde{G}_{Y| X}^{\ast}(\cdot| x)
\bigr)
\quad\text{in }\ell^\infty(T).
\end{equation*}

\end{proof}
\vspace{0.5cm}

\vspace{0.5cm}
\begin{proof}[\textnormal{\textbf{Proof of Proposition \ref{proposition:AQTE}}}]

\textbf{(a)} 
By the definition of the AQTE, 
we have 
\begin{equation*}
  \theta(0, 1)
  =
  \int_{0}^{1} F_{Y(1)}^{-1}(u)\, du
  -
  \int_{0}^{1} F_{Y(0)}^{-1}(u)\, du.   
\end{equation*}
The quantile function satisfies that  
$u > F_{Y(w)}(t)$
if and only if 
$F_{Y(w)}^{-1}(u) > t$
for every $u \in (0, 1)$ and $t \in \R$.
By Fubini's theorem, we have that for each $w \in \{0,1\}$,
\begin{equation*}
\int_{0}^{1} F_{Y(w)}^{-1}(u)\, du
= \int_{0}^{\infty} \big (1 - F_{Y(w)}(t) \big)\, dt
- \int_{-\infty}^{0} F_{Y(w)}(t)\, dt
= \E[Y(w)],
\end{equation*}
where the last equality uses the tail integral formula for the
expectation.
Thus, the desired result follows. 

\textbf{(b)}
For 
$\tau_{\ell} , \tau_{u} \in (0,1)$ with
$\tau_{\ell}  < \tau_{u}$,  
the AQTE can be written 
as the difference of two difference quotients: 
\begin{equation*}
\theta(\tau_{\ell}, \tau_{u})
= \frac{G_{Y(1)}^{\ast}(\tau_{u}) - G_{Y(1)}^{\ast}(\tau_{\ell})}{\tau_{u} - \tau_{\ell}}
- \frac{G_{Y(0)}^{\ast}(\tau_{u}) - G_{Y(0)}^{\ast}(\tau_{\ell})}{\tau_{u} - \tau_{\ell}}.
\end{equation*}
By Theorem 24.1 of \cite{rockafellar1997convex},
as $(\tau_{\ell}, \tau_{u}) \to (\tau, \tau)$ with
$\tau_{\ell} < \tau_{u}$, every limit point of the difference
quotient of $G_{Y(w)}^{\ast}$ lies in $\partial G_{Y(w)}^{\ast}(\tau)$
for each $w \in \{0,1\}$.
It follows that every limit point of $\theta(\tau_{\ell}, \tau_{u})$
belongs to the Minkowski difference
\begin{equation*}
\partial G_{Y(1)}^{\ast}(\tau) \ominus \partial G_{Y(0)}^{\ast}(\tau)
= [F_{Y(1)}^{-1}(\tau) - F_{Y(0)}^{-1}(\tau{\scalebox{0.7}{$+$}}),\,
F_{Y(1)}^{-1}(\tau{\scalebox{0.7}{$+$}}) - F_{Y(0)}^{-1}(\tau)],
\end{equation*}
where $A \ominus B := \{a - b : a \in A,\ b \in B\}$ for
$A, B \subset \R$. When $F_{Y(w)}^{-1}$ is continuous at $\tau$ for
each $w \in \{0, 1\}$, both subdifferentials reduce to singletons,
and this interval collapses to the single point
$F_{Y(1)}^{-1}(\tau) - F_{Y(0)}^{-1}(\tau)$.
\end{proof}
\vspace{0.5cm}

\vspace{0.5cm}
\begin{proof}[\textnormal{\textbf{Proof of Theorem \ref{theorem:AQTE}}}]

Under Assumption \ref{assump:random-assignment}(c), 
$(Y(0), Y(1)) \independent W$. 
Thus, for each $w \in \{0, 1\}$ with $\Pr(W = w) >
0$,
\begin{equation*}
  \Pr\{Y \le y | W = w\}
  =
  \Pr\{Y(W) \le y | W = w\}
  =
  \Pr\{Y(w) \le y\}
  =
  F_{Y(w)}(y),
\end{equation*}
where the first equality uses $Y = Y(W)$ from Assumption
\ref{assump:random-assignment}(a). The marginal distribution
function $F_{Y(w)}$ is therefore identified from the joint
distribution of $(W, Y)$. Since $\E[|Y(w)|] < \infty$ under
Assumption \ref{assump:random-assignment}(b), Proposition
\ref{proposition:L2-min-1}(b) gives
\begin{equation*}
  G_{Y(w)}^{\ast}(\tau)
  =
  \int_{0}^{\tau} F_{Y(w)}^{-1}(u)\, du
  \quad \text{for } \tau \in [0, 1],
\end{equation*}
with 
$G_{Y(w)}^{\ast}(0) = 0$ and $G_{Y(w)}^{\ast}(1) = \E[Y(w)]$, both
of which are identified under Assumption 
\ref{assump:random-assignment}(b)-(c).
Thus,  $G_{Y(w)}^{\ast}(\cdot)$ is identified. For any $\tau_{\ell}, \tau_{u}
\in [0, 1]$ with $\tau_{\ell} < \tau_{u}$, we can show
\begin{equation*}
  \theta(\tau_{\ell}, \tau_{u})
  =
  \frac{1}{\tau_{u} - \tau_{\ell}}
  \big \{
    \big (      
    G_{Y(1)}^{\ast}(\tau_{u}) - G_{Y(1)}^{\ast}(\tau_{\ell})
    \big ) 
    -
    \big (  
    G_{Y(0)}^{\ast}(\tau_{u}) - G_{Y(0)}^{\ast}(\tau_{\ell})
    \big)
  \big \},
\end{equation*}
which yields the identification of $\theta(\tau_{\ell}, \tau_{u})$.
\end{proof}
\vspace{0.5cm}

\vspace{0.5cm}
\begin{proof}[\textnormal{\textbf{Proof of Theorem
\ref{theorem:AQTE-asymp}}}]

For each $w \in \{0,1\}$, the estimator
$\widehat{G}_{Y(w)}(y)$ is the treatment-group average defined in
Section \ref{sec:treatment}. Under random assignment,
$\E[(y-Y)_{+}| W=w]=G_{Y(w)}(y)$. The second-moment condition in
Assumption \ref{assump:aqte-regularity}(b), together with the same
Donsker argument as in Theorem \ref{theorem:aym}, therefore gives
\begin{equation*}
  \sqrt{n}\big( \widehat{G}_{Y(w)}(\cdot) - G_{Y(w)}(\cdot) \big)
  \rightsquigarrow
  \mathbb{B}_{w}(\cdot)
  \quad \text{in } \ell^{\infty}(\Rex), 
\end{equation*}
where $\mathbb{B}_{w}(\cdot)$
is a zero-mean Gaussian process with covariance kernel
$\Sigma_{w}(y_{1}, y_{2})$. The processes
$\mathbb{B}_{0}$ and $\mathbb{B}_{1}$ are independent because their
influence functions $\1\{W = w\}[(y - Y)_{+} - G_{Y(w)}(y)] /
\Pr(W = w)$ for $w = 0$ and $w = 1$ are supported on disjoint
events.

The fixed-index conjugate argument in Theorem
\ref{theorem:L-weak}, applied directly to
$\widehat{G}_{Y(w)}$ and the finite collection
$\{\tau_{\ell},\tau_{u}\}\subset[0,1]$, yields
\begin{equation*}
  \big\{
  \sqrt{n}\big( \widehat{G}_{Y(w)}^{\ast}(\tau) -
  G_{Y(w)}^{\ast}(\tau) \big)
  \big\}_{\tau \in \{\tau_{\ell}, \tau_{u}\}}
  \rightsquigarrow
  \big\{ \mathbb{Z}_{w}(\tau) \big\}_{\tau \in \{\tau_{\ell}, \tau_{u}\}},
\end{equation*}
where 
$
  \mathbb{Z}_{w}(\tau)
  =
  \lim_{\eta \downarrow 0}
  \sup_{y \in \partial_{\eta} G_{Y(w)}^{\ast}(\tau)}
  \big\{ -\mathbb{B}_{w}(y) \big\}
$
for each $w \in \{0,1\}$.

For $\tau \in (0, 1)$, the sets $\partial_{\eta}
G_{Y(w)}^{\ast}(\tau)$ are nested intervals decreasing to
$\partial G_{Y(w)}^{\ast}(\tau)$ as $\eta \downarrow 0$, and the
$\rho$-uniform continuity of $\mathbb{B}_{w}$ on $(\Rex, \rho)$
yields 
$\mathbb{Z}_{w}(\tau) = \sup_{y \in \partial G_{Y(w)}^{\ast}(\tau)}
\{-\mathbb{B}_{w}(y)\}$. 
At $\tau = 0$, 
$\partial_{\eta}G_{Y(w)}^{\ast}(0)$ 
is a half-line whose upper endpoint decreases to
$\inf \mathrm{supp}(Y(w))$. If the support is unbounded below, the
sets shrink to $\{-\infty\}$ in $(\Rex, \rho)$ and
$\mathbb{B}_{w}(-\infty) = 0$, while if bounded below, $\Sigma_{w}$
vanishes on the limiting half-line. It follows that $\mathbb{Z}_{w}(0) = 0$ in either case.
At $\tau = 1$, $\partial_{\eta} G_{Y(w)}^{\ast}(1)$ is a
half-line whose lower endpoint increases to
$\sup \mathrm{supp}(Y(w))$. 
If the support is unbounded above, the
sets shrink to $\{+\infty\}$, while if bounded above,
$\mathbb{B}_{w}$ is almost surely constant on the limiting
half-line. In either case 
$\mathbb{Z}_{w}(1) = -\mathbb{B}_{w}(+\infty)$,
a centered Gaussian with variance $\var(Y(w)) / \Pr(W = w)$.

The AQTE estimator $\widehat{\theta}(\tau_{\ell}, \tau_{u})$ in
\eqref{eq:aqte-hat} is a linear combination of the four conjugate
estimators $\widehat{G}_{Y(w)}^{\ast}(\tau)$ for $(w, \tau) \in
\{0, 1\} \times \{\tau_{\ell}, \tau_{u}\}$. 
Because the processes for the treatment and control groups converge
jointly in $\ell^{\infty}(\Rex)^{2}$,
the argument in the proof of
Theorem \ref{theorem:L-weak} yields joint weak convergence
of these four estimators.
The continuous mapping theorem gives
\begin{equation*}
  \sqrt{n}\big(
    \widehat{\theta}(\tau_{\ell}, \tau_{u})
    - \theta(\tau_{\ell}, \tau_{u})
  \big)
  \rightsquigarrow
  \frac{1}{\tau_{u} - \tau_{\ell}}
  \big\{
    \big( \mathbb{Z}_{1}(\tau_{u}) - \mathbb{Z}_{1}(\tau_{\ell}) \big)
    - \big( \mathbb{Z}_{0}(\tau_{u}) - \mathbb{Z}_{0}(\tau_{\ell}) \big)
  \big\},
\end{equation*}
which establishes the limit distribution.

  The delta-method bootstrap based on the sample
  $\eta_n$-subdifferentials
  $\partial_{\eta_n}\widehat{G}_{Y(w)}^{\ast}(\tau)$ for
  $(w,\tau)\in\{0,1\}\times\{\tau_{\ell},\tau_{u}\}$, followed by the
  linear functional in \eqref{eq:aqte-hat}, consistently estimates the
  limit law, provided that $\eta_n\to0$ and
  $\sqrt{n}\eta_n\to\infty$.
\end{proof}
\vspace{0.5cm} 


\newpage
\section*{Appendix B: Conditional Treatment Effect}
\label{appendix:conditional}

\setcounter{equation}{0}
\setcounter{theorem}{0}
\setcounter{proposition}{0}
 \setcounter{assumption}{0}

\renewcommand{\theequation}{B.\arabic{equation}}
\renewcommand{\theHequation}{B.\arabic{equation}}
\renewcommand{\thetheorem}{B.\arabic{theorem}}
\renewcommand{\theproposition}{B.\arabic{proposition}}
\renewcommand{\theHproposition}{B.\arabic{proposition}}
\renewcommand{\theassumption}{B\arabic{assumption}}
\renewcommand{\theHassumption}{B\arabic{assumption}}
\renewcommand{\theHtheorem}{B.\arabic{theorem}}

This appendix develops the conditional version of the average
quantile treatment effect. The conditional AQTE measures
distributional treatment effects within the subpopulation defined
by a covariate value. We identify the conditional AQTE under
unconfoundedness of the treatment assignment given those
covariates, combined with a linear specification of the integrated
potential-outcome distribution function via the ReLU regression of
Section \ref{sec:relu}. The notation and framework of Section
\ref{sec:treatment} are maintained throughout, with covariates $X$
augmenting the observed data.

\subsection*{B.1 Setup and the Conditional AQTE}

Let $X \in \X \subseteq \R^{q}$ be a vector of pre-treatment
covariates realized before the assignment of $W$. We observe an
i.i.d.\ sample $\{(W_{i}, X_{i}, Y_{i})\}_{i=1}^{n}$ drawn from the
joint distribution of $(W, X, Y)$.

For each treatment status $w \in \{0, 1\}$, the conditional
distribution function and conditional quantile function of the
potential outcome $Y(w)$ given $X$ are
\begin{equation*}
  F_{Y(w) | X}(y | X)
  := \Pr\{Y(w) \le y | X\},
  \qquad
  F_{Y(w) | X}^{-1}(u | X)
  := \inf\{y \in \R : F_{Y(w) | X}(y | X) \ge u\},
\end{equation*}
for $y \in \R$ and $u \in (0, 1)$. The integrated conditional
distribution function and its convex conjugate are
\begin{equation*}
  G_{Y(w) | X}(y | X)
  := \int_{-\infty}^{y} F_{Y(w) | X}(s | X)\, ds,
  \qquad
  G_{Y(w) | X}^{\ast}(\tau | X)
  := \sup_{y \in \R}\big\{ \tau y - G_{Y(w) | X}(y | X) \big\}.
\end{equation*}
The first object is convex in $y$ for each fixed value of $X$, since
the integrand $F_{Y(w) | X}(\cdot | X)$ is nondecreasing.

For probability levels $\tau_{\ell}, \tau_{u} \in [0, 1]$ with
$\tau_{\ell} < \tau_{u}$, the conditional AQTE given $X = x$ is
defined as
\begin{equation}
  \label{eq:caqte}
  \theta(\tau_{\ell}, \tau_{u}, x)
  :=
  \frac{1}{\tau_{u} - \tau_{\ell}}
  \bigg\{
  \int_{\tau_{\ell}}^{\tau_{u}} F_{Y(1) | X}^{-1}(u | x)\, du
  -
  \int_{\tau_{\ell}}^{\tau_{u}} F_{Y(0) | X}^{-1}(u | x)\, du
  \bigg\}.
\end{equation}

The conditional AQTE encompasses standard conditional
treatment-effect parameters as special cases. The next proposition
records the analogue of Proposition \ref{proposition:AQTE} for the
conditional case.

\vspace{0.3cm}
\begin{proposition}
  \label{proposition:AQTE-conditional}
  Suppose $\E[|Y(w)|] < \infty$ for each $w \in \{0, 1\}$. Then for
  almost every $x \in \X$,
  \begin{enumerate}[label=\textnormal{(\roman*)}, topsep=3pt, itemsep=3pt, align=left, leftmargin=1em, labelsep=0.5em]
    \item $\theta(0, 1, x) = \E[Y(1) - Y(0) | X = x]$;
    \item for any $\tau \in (0, 1)$, every limit point of
      $\theta(\tau_{\ell}, \tau_{u}, x)$ as $(\tau_{\ell},
      \tau_{u}) \to (\tau, \tau)$ with $\tau_{\ell} < \tau_{u}$
      belongs to $[F_{Y(1) | X}^{-1}(\tau | x) -
      F_{Y(0) | X}^{-1}(\tau{\scalebox{0.7}{$+$}} | x),\,
      F_{Y(1) | X}^{-1}(\tau{\scalebox{0.7}{$+$}} | x) -
      F_{Y(0) | X}^{-1}(\tau | x)]$. If $F_{Y(w) | X}^{-1}(\cdot | x)$ is
      continuous at $\tau$ for each $w \in \{0, 1\}$, the limit is
      unique and equals $F_{Y(1) | X}^{-1}(\tau | x) -
      F_{Y(0) | X}^{-1}(\tau | x)$.
  \end{enumerate}
\end{proposition}
\vspace{0.3cm}

As in the unconditional case, the endpoints are admissible because 
$G_{Y(w) | X}^{\ast}(0 | x) = 0$ and $G_{Y(w) | X}^{\ast}(1 | x) = \E[Y(w) | X = x]$. 
The proof is the conditional analogue of the proof of Proposition
\ref{proposition:AQTE}  and is omitted.

\subsection*{B.2 Identification}

For each fixed $y \in \R$ and $w \in \{0, 1\}$, we specify the
ReLU regression for the potential outcome as
\begin{equation}
  \label{eq:relu-potential}
  (y - Y(w))_{+}
  =
  b(X)^{\top} \gamma_{0}(y)
  + w \cdot b(X)^{\top} \delta_{0}(y)
  + \epsilon_{w}(y),
\end{equation}
where $b: \X \to \R^{r}$ is a known basis function with the first
element of $b(X)$ a constant, $\gamma_{0}(y), \delta_{0}(y) \in
\R^{r}$ are unknown coefficient vectors, and $\epsilon_{w}(y)$ is a
mean-zero error term. Evaluating \eqref{eq:relu-potential} at $w=W$
and using $Y=Y(W)$ yields the observed ReLU regression on $b(X)$ and
$W b(X)$. We define the estimators of $\gamma_{0}(y)$ and
$\delta_{0}(y)$ directly below.

Identification of the conditional AQTE requires unconfoundedness of
the treatment assignment given the covariates and a linear
specification of the integrated potential-outcome distribution
function.

\vspace{0.3cm}
\begin{assumption}
  \label{assump:unconfoundedness}
  The data-generating process satisfies:
  \begin{enumerate}[label=(\alph*), noitemsep, topsep=0pt]
    \item[\textnormal{(a)}] $Y = Y(W)$ almost surely.
    \item[\textnormal{(b)}] $\E[|Y(w)|] < \infty$ for each $w \in
      \{0, 1\}$.
    \item[\textnormal{(c)}] $W \independent (Y(0), Y(1)) \mid X$ and
      $0 < \Pr(W = 1 | X) < 1$ almost surely.
    \item[\textnormal{(d)}] $\E[\epsilon_{w}(y) | X] = 0$ for each
      $(y, w) \in \R \times \{0, 1\}$.
    \item[\textnormal{(e)}] $\E[Y^{2}] < \infty$,
      $\E[\|b(X)\|^{2}] < \infty$, and the population Gram matrix
      for $[b(X)^{\top}, W b(X)^{\top}]^{\top}$ is positive definite.
  \end{enumerate}
\end{assumption}
\vspace{0.3cm}

Conditions (a) and (b) of Assumption \ref{assump:unconfoundedness}
match conditions (a) and (b) of Assumption
\ref{assump:random-assignment}. Condition (c) is the standard
unconfoundedness assumption \citep{rubin1980randomization}, which
renders the treatment ignorable given $X$ and imposes overlap.
Condition (d) places the conditional mean-zero restriction on the
error term in \eqref{eq:relu-potential}. Under (d), the linear
specification \eqref{eq:relu-potential} is correctly specified for
$\E[(y - Y(w))_{+} | X]$. 
Condition (e) imposes the moment and rank conditions for the
regression on $b(X)$ and $W b(X)$. Under overlap, its rank condition
reduces to positive-definiteness of $\E[b(X)b(X)^{\top}]$.

\vspace{0.3cm}
\begin{theorem}
  \label{theorem:AQTE-conditional}
  Suppose Assumption \ref{assump:unconfoundedness} holds. 
  Then,
  for any $\tau_{\ell}, \tau_{u} \in [0, 1]$
  with $\tau_{\ell} < \tau_{u}$ and almost every $x \in \X$, the conditional AQTE
  $\theta(\tau_{\ell}, \tau_{u}, x)$ defined in \eqref{eq:caqte} is
  identified from the joint distribution of $(W, X, Y)$.
\end{theorem}
\begin{proof}

Fix $x \in \X$ and $w \in \{0, 1\}$. Under
Assumption \ref{assump:unconfoundedness}(c), the conditional
distribution of $Y(w)$ given $X$ coincides with the conditional
distribution of $Y$ given $W = w$ and $X$:
\begin{equation*}
  \Pr\{Y \le y | W = w, X = x\}
  =
  \Pr\{Y(W) \le y | W = w, X = x\}
  =
  \Pr\{Y(w) \le y | X = x\},
\end{equation*}
where the first equality uses Assumption
\ref{assump:unconfoundedness}(a), the second substitutes $W = w$
on the conditioning event and uses the conditional
independence in (c). The conditional distribution function
$F_{Y(w) | X}(\cdot | x)$ is therefore identified from the joint
distribution of $(W, X, Y)$.

Taking the conditional expectation given $X = x$ in
\eqref{eq:relu-potential} and using Assumption
\ref{assump:unconfoundedness}(d), we have
\begin{equation*}
  \E[(y - Y(w))_{+} | X = x]
  =
  b(x)^{\top} \gamma_{0}(y) + w \cdot b(x)^{\top} \delta_{0}(y).
\end{equation*}
By Proposition \ref{proposition:L2-min-1}(a), the left-hand side
equals $G_{Y(w) | X}(y | x)$. Under condition (e), the population
least-squares criterion for the regression on $b(X)$ and $W b(X)$
has the unique minimizer $(\gamma_{0}(y),\delta_{0}(y))$. Hence
$G_{Y(w) | X}(y | x)$ is identified for
each $w \in \{0, 1\}$ and each $y \in \R$.

By Proposition \ref{proposition:L2-min-1}(b),
\begin{equation*}
  G_{Y(w) | X}^{\ast}(\tau | x)
  =
  \int_{0}^{\tau} F_{Y(w) | X}^{-1}(u | x)\, du
\end{equation*}
for each $\tau \in (0, 1)$, and $G_{Y(w) | X}^{\ast}(\cdot | x)$ is
therefore identified. For any $\tau_{\ell}, \tau_{u} \in [0, 1]$
with $\tau_{\ell} < \tau_{u}$, the conditional AQTE in
\eqref{eq:caqte} satisfies
\begin{equation*}
  \theta(\tau_{\ell}, \tau_{u}, x)
  =
  \frac{1}{\tau_{u} - \tau_{\ell}}
  \big(
    G_{Y(1) | X}^{\ast}(\tau_{u} | x) - G_{Y(1) | X}^{\ast}(\tau_{\ell} | x)
    -
    G_{Y(0) | X}^{\ast}(\tau_{u} | x) + G_{Y(0) | X}^{\ast}(\tau_{\ell} | x)
  \big),
\end{equation*}
which yields the identification of $\theta(\tau_{\ell}, \tau_{u},
x)$.
\end{proof}
\vspace{0.3cm}

\subsection*{B.3 Estimation and Asymptotic Properties}

For each $y \in \R$, we define
$(\hat{\gamma}(y),\hat{\delta}(y))$ as the minimizer over
$\R^{r}\times\R^{r}$ of
\begin{equation*}
  \frac{1}{n}\sum_{i=1}^{n}
  \big\{(y-Y_i)_{+}
  -b(X_i)^{\top}\gamma
  -W_i b(X_i)^{\top}\delta\big\}^{2}.
\end{equation*}
We estimate the conditional integrated distribution function as
\begin{equation*}
  \widehat{G}_{Y(w) | X}(y | x)
  :=
  b(x)^{\top} \hat{\gamma}(y) + w \cdot b(x)^{\top}
  \hat{\delta}(y).
\end{equation*}
Its convex conjugate is
\begin{equation*}
  \widehat{G}_{Y(w) | X}^{\ast}(\tau | x)
  :=
  \sup_{y \in \R}
  \big\{ \tau y - \widehat{G}_{Y(w) | X}(y | x) \big\},
\end{equation*}
and the conditional AQTE estimator is
\begin{equation}
  \label{eq:caqte-hat}
  \widehat{\theta}(\tau_{\ell}, \tau_{u}, x)
  :=
  \frac{1}{\tau_{u} - \tau_{\ell}}
  \big[
    \widehat{G}_{Y(1) | X}^{\ast}(\tau_{u} | x)
    - \widehat{G}_{Y(1) | X}^{\ast}(\tau_{\ell} | x)
    - \widehat{G}_{Y(0) | X}^{\ast}(\tau_{u} | x)
    + \widehat{G}_{Y(0) | X}^{\ast}(\tau_{\ell} | x)
  \big].
\end{equation}

The asymptotic distribution of $\widehat{\theta}(\tau_{\ell},
\tau_{u}, x)$ follows from the asymptotic theory of Section
\ref{sec:estimation} specialized to the fully interacted
specification and combined with the delta method for Hadamard directionally
differentiable maps.
We define,
for each 
  $x \in \X$,
  $w \in \{0,1\}$
  and  
  $\tau \in [0,1]$,
  \begin{equation*}
    \mathbb{Z}_{w}(\tau | x)
    :=
    \lim_{\eta \downarrow 0}
    \sup_{y \in \partial_{\eta} G_{Y(w) | X}^{\ast}(\tau | x)}
    \big\{ -\mathbb{B}_{w}(y | x) \big\}, 
  \end{equation*}
  where the conditional Gaussian process is
  \begin{equation*}
    \mathbb{B}_{w}(y | x)
    :=
    b(x)^{\top} \mathbb{B}_{\gamma}(y)
    + w \cdot b(x)^{\top} \mathbb{B}_{\delta}(y).
  \end{equation*}
Here, $\mathbb{B}_{\gamma}$ and $\mathbb{B}_{\delta}$ are the joint
Gaussian limits of the coefficient estimators for $b(X)$ and
$W b(X)$, respectively.

\vspace{0.3cm}
\begin{theorem}
  \label{theorem:caqte-asymp}
  Suppose that Assumption \ref{assump:unconfoundedness} holds,
  the sample is i.i.d., and
  $\E[\|b(X)Y\|^{2}] < \infty$ and
  $\E[\|b(X)\|^{4}] < \infty$.
  Then, the following statements hold.
  \begin{enumerate}[label=\textnormal{(\alph*)}, topsep=3pt, itemsep=3pt, align=left, leftmargin=1em, labelsep=0.5em]
    \item 
    For any $\tau_{\ell}, \tau_{u} \in [0, 1]$ 
    with
  $\tau_{\ell} < \tau_{u}$ and almost every 
    $x \in \X$, 
  \begin{equation*}
    \sqrt{n}\big( \widehat{\theta}(\tau_{\ell}, \tau_{u}, x)
      - \theta(\tau_{\ell}, \tau_{u}, x) \big) \rightsquigarrow
      \dfrac{1}{\tau_{u} - \tau_{\ell}}
      \big \{  
        \big ( \mathbb{Z}_{1}(\tau_{u} | x) 
          - 
          \mathbb{Z}_{1}(\tau_{\ell} | x)
        \big ) 
      - \big ( \mathbb{Z}_{0}(\tau_{u} | x) - \mathbb{Z}_{0}(\tau_{\ell} | x)
      \big)  
      \big \}   .
  \end{equation*}
    \item The delta-method bootstrap 
      based on the empirical conditional subdifferential
      $\partial_{\eta_{n}} 
      \widehat{G}_{Y(w) | X }^{\ast}(\tau|x)$ 
      consistently
      estimates the limit law in part (a), 
      provided that $\eta_n\to0$ and $\sqrt{n}\eta_n\to\infty$.
  \end{enumerate}
\end{theorem}
\begin{proof}
 
  Applying an argument similar to that in the proof of Theorem \ref{theorem:aym}, we can show the following  joint convergence:
  \begin{equation*}
    \sqrt{n}\big(
      \hat{\gamma}(\cdot)-\gamma_{0}(\cdot),
      \hat{\delta}(\cdot)-\delta_{0}(\cdot)
    \big)
    \rightsquigarrow
    \big(
      \mathbb{B}_{\gamma}(\cdot),
      \mathbb{B}_{\delta}(\cdot)
    \big)
    \quad \text{in } \ell^{\infty}(\Rex)^{2r}.
  \end{equation*}
  This result with the continuous mapping theorem
  yields, for almost every $x \in \X$, 
  \begin{equation*}
    \sqrt{n}\big( \widehat{G}_{Y(w) | X}(\cdot | x)
    - G_{Y(w) | X}(\cdot | x) \big)
    \rightsquigarrow
    \mathbb{B}_{w}(\cdot | x)
    \quad \text{in } \ell^{\infty}(\Rex),
  \end{equation*}
where 
$\mathbb{B}_{w}(\cdot | x) := b(x)^{\top} \mathbb{B}_{\gamma}(\cdot) +
w \cdot b(x)^{\top} \mathbb{B}_{\delta}(\cdot)$.
Applying 
Theorem \ref{theorem:L-weak} 
with the conditional function $G_{Y(w) | X}(\cdot | x)$ and the finite collection
$\{\tau_{\ell}, \tau_{u}\} \subset [0, 1]$, 
we can show that, for each
$w \in \{0, 1\}$,
\begin{equation*}
  \big \{ 
  \sqrt{n}\big( \widehat{G}_{Y(w) | X}^{\ast}(\tau | x)
  - G_{Y(w) | X}^{\ast}(\tau | x) \big)
  \big \}_{\tau \in \{\tau_\ell, \tau_u\}}
  \rightsquigarrow
  \big \{ 
  \mathbb{Z}_{w}(\tau | x)
  \big \}_{\tau \in \{\tau_\ell, \tau_u\}} ,
\end{equation*}
where 
$ \mathbb{Z}_{w}(\tau | x)
  =
  \lim_{\eta \downarrow 0}
  \sup_{y \in \partial_{\eta} G_{Y(w) | X}^{\ast}(\tau | x)}
  \{ -\mathbb{B}_{w}(y | x) \}.
$
The endpoint cases $\tau \in \{0, 1\}$ carry over from the proof of
Theorem \ref{theorem:AQTE-asymp} conditionally on $X = x$, so that
$\mathbb{Z}_{w}(0 | x) = 0$ 
and 
$\mathbb{Z}_{w}(1 | x) = -\mathbb{B}_{w}(+\infty | x)$.
The conditional AQTE estimator $\widehat{\theta}(\tau_{\ell},
\tau_{u}, x)$ in \eqref{eq:caqte-hat} is a linear combination of
the four conjugate values $\widehat{G}_{Y(w) | X}^{\ast}(\tau | x)$.
Joint weak convergence of the four objects, combined with the
continuous mapping theorem, yields part (a). Part (b) follows from
the delta-method bootstrap consistency result of
\cite{fang2019inference} applied to the linear functional in
\eqref{eq:caqte-hat}.
\end{proof}
\vspace{0.3cm}

In the above theorem, 
when $b(X)$ is bounded, as with the saturated stratum dummies in
Section \ref{sec:application}, the moment condition reduces to
$\E[Y(w)^{2}] < \infty$ for each $w \in \{0, 1\}$.
The covariance kernel of the conditional Gaussian process
$\mathbb{B}_{w}(\cdot | x)$ is determined by Theorem
\ref{theorem:aym} applied to the regression on $b(X)$ and $W b(X)$. 
  For each $w \in \{0,1\}$ and $\tau \in (0,1)$, the component limit
  $\mathbb{Z}_w(\tau | x)$ is centered Gaussian when
  $
  F_{Y(w) | X}^{-1}(\tau | x)
  =
  F_{Y(w) | X}^{-1}(\tau{\scalebox{0.7}{$+$}} | x),
  $
  so that the conjugate problem has a unique maximizer. It may be
  non-Gaussian when the inequality is strict. This case corresponds to
  a flat segment of the conditional distribution function at level
  $\tau$, or equivalently a gap in the conditional outcome support; an
  atom alone does not generate non-Gaussianity. At such probability
  indices, the standard nonparametric bootstrap is generally
  inconsistent, and the delta-method procedure in part (b) is required.

\subsection*{B.4 Marginalization over the Conditioning Variables}

The marginal integrated distribution function is obtained by
integrating the conditional function over the distribution of $X$.
Let
\begin{equation*}
  \mu_b:=\E[b(X)],
  \qquad
  \bar{b}_n:=\frac{1}{n}\sum_{i=1}^n b(X_i).
\end{equation*}
Under Assumption \ref{assump:unconfoundedness}, iterated expectations
and \eqref{eq:relu-potential} give, for each $w\in\{0,1\}$,
\begin{equation*}
  G_{Y(w)}(y)
  =
  \E\big[G_{Y(w)| X}(y | X)\big]
  =
  \mu_b^{\top}\{\gamma_0(y)+w\delta_0(y)\}.
\end{equation*}
The corresponding marginalized estimator used in Section
\ref{sec:application} is
\begin{equation}
  \label{eq:marginalized-G}
  \widehat{G}_{Y(w)}(y)
  :=
  \frac{1}{n}\sum_{i=1}^n
  \widehat{G}_{Y(w)| X}(y | X_i)
  =
  \bar{b}_n^{\top}\{\hat{\gamma}(y)+w\hat{\delta}(y)\}.
\end{equation}
Its conjugate $\widehat{G}_{Y(w)}^{\ast}$ and the marginal AQTE estimator
$\widehat{\theta}(\tau_\ell,\tau_u)$ are defined as in
\eqref{eq:aqte-hat}.

To state the limit process, let $\mathbb A_w$ denote the centered
Gaussian process that arises jointly with
$(\mathbb B_\gamma,\mathbb B_\delta)$ as the weak limit of
\begin{equation*}
  \mathbb G_n h_{w,y},
  \qquad
  h_{w,y}(X)
  :=
  b(X)^{\top}\{\gamma_0(y)+w\delta_0(y)\}-(y)_+,
  \qquad y\in\R,
\end{equation*}
where $\mathbb G_n:=\sqrt{n}(P_n-P)$. We extend
$h_{w,\cdot}$ to $\Rex$ by its limits as $y\to\pm\infty$. Define
\begin{equation*}
  \bar{\mathbb B}_w(y)
  :=
  \mu_b^{\top}\{\mathbb B_\gamma(y)+w\mathbb B_\delta(y)\}
  +\mathbb A_w(y).
\end{equation*}
Because these processes are joint limits from the same empirical
process, the covariance kernel of $\bar{\mathbb B}_w$ includes the
cross-covariance between coefficient estimation and empirical
marginalization.

For $\tau\in[0,1]$, define
\begin{equation*}
  \bar{\mathbb{Z}}_w(\tau)
  :=
  \lim_{\eta \downarrow0}
  \sup_{y\in\partial_{\eta} G_{Y(w)}^{\ast}(\tau)}
  \{-\bar{\mathbb B}_w(y)\}.
\end{equation*}

\vspace{0.3cm}
\begin{theorem}
  \label{theorem:marginalized-asymp}
  Suppose that the conditions of Theorem
  \ref{theorem:caqte-asymp} hold. Then, the following statements hold.
  \begin{enumerate}[label=\textnormal{(\alph*)}, topsep=3pt, itemsep=3pt, align=left, leftmargin=1em, labelsep=0.5em]
    \item Jointly in $w\in\{0,1\}$,
      \begin{equation*}
        \big\{
        \sqrt{n}\big(\widehat{G}_{Y(w)}(\cdot)-G_{Y(w)}(\cdot)\big)
        \big\}_{w\in\{0,1\}}
        \rightsquigarrow
        \big\{\bar{\mathbb B}_w(\cdot)\big\}_{w\in\{0,1\}}
        \quad\text{in }\ell^\infty(\Rex)^2.
      \end{equation*}
    \item For fixed $\tau_\ell,\tau_u\in[0,1]$ with
      $\tau_\ell<\tau_u$,
      \begin{equation*}
        \sqrt{n}\big(
          \widehat{\theta}(\tau_\ell,\tau_u)
          -\theta(\tau_\ell,\tau_u)
        \big)
        \rightsquigarrow
        \frac{
          \{\bar{\mathbb{Z}}_1(\tau_u)
          -
          \bar{\mathbb{Z}}_1(\tau_\ell)\}
          -\{
            \bar{\mathbb{Z}}_0(\tau_u)
            -
            \bar{\mathbb{Z}}_0(\tau_\ell)\}
        }{\tau_u-\tau_\ell}.
      \end{equation*}

    \item The delta-method bootstrap based on the sample
    $\eta_n$-subdifferentials
    $\partial_{\eta_n}\widehat{G}_{Y(w)}^{\ast}(\tau)$ for
    $(w,\tau)\in\{0,1\}\times\{\tau_{\ell},\tau_{u}\}$
    consistently estimates the limit law in part (b), provided that
    $\eta_n\to0$, $\sqrt{n}\eta_n\to\infty$, and both the coefficient
    estimators and the empirical marginalization in
    \eqref{eq:marginalized-G} are recomputed using the same bootstrap
    weights.
  \end{enumerate}
\end{theorem}
\begin{proof}
  The proof proceeds in two steps and uses the joint convergence in
the proof of Theorem \ref{theorem:caqte-asymp}.

\textbf{Step 1.}
For each $w\in\{0,1\}$ and $y\in\R$, adding and subtracting
$\mu_b^{\top}\{\hat{\gamma}(y)+w\hat{\delta}(y)\}$ in
\eqref{eq:marginalized-G} gives
\begin{align*}
  &\sqrt{n}\big\{\widehat{G}_{Y(w)}(y)-G_{Y(w)}(y)\big\}\\
  &\quad=
  \mu_b^{\top}\sqrt{n}
  \big[\{\hat{\gamma}(y)-\gamma_0(y)\}
  +w\{\hat{\delta}(y)-\delta_0(y)\}\big]\\
  &\qquad+
  \mathbb G_n\big[
    b(X)^{\top}\{\gamma_0(y)+w\delta_0(y)\}
  \big]
  +R_{n,w}(y),
\end{align*}
where
\begin{equation*}
  R_{n,w}(y)
  :=
  \sqrt{n}(\bar{b}_n-\mu_b)^{\top}
  \big[\{\hat{\gamma}(y)-\gamma_0(y)\}
  +w\{\hat{\delta}(y)-\delta_0(y)\}\big].
\end{equation*}
The coefficient-process convergence in the proof of
Theorem \ref{theorem:caqte-asymp} and
$\sqrt{n}(\bar{b}_n-\mu_b)=O_p(1)$ imply
\begin{equation*}
  \max_{w\in\{0,1\}}
  \sup_{y\in\Rex}|R_{n,w}(y)|=o_p(1).
\end{equation*}

Because $(y)_+$ is nonrandom,
\begin{equation*}
  \mathbb G_n\big[
    b(X)^{\top}\{\gamma_0(y)+w\delta_0(y)\}
  \big]
  =
  \mathbb G_n h_{w,y}.
\end{equation*}
The class $\{h_{w,y}:y\in\Rex,\ w\in\{0,1\}\}$ is Donsker under
the moment conditions of Theorem \ref{theorem:caqte-asymp}. Indeed,
because the first element of $b(X)$ equals one, the subtraction of
$(y)_+$ in the definition of $h_{w,y}$ removes the unbounded
intercept component. The
remaining coefficient path is bounded on $\Rex$ by the endpoint
argument in Appendix A. Thus, the same finite-dimensional
Lipschitz-envelope argument used in the proof of Theorem
\ref{theorem:aym} applies. Hence the empirical processes
$\mathbb G_n h_{w,\cdot}$ converge jointly with the coefficient
processes to $\mathbb A_w$, $\mathbb B_\gamma$, and
$\mathbb B_\delta$. The expansion above and the continuous mapping
theorem yield
\begin{equation*}
  \big\{
    \sqrt{n}(\widehat{G}_{Y(w)}-G_{Y(w)})
  \big\}_{w\in\{0,1\}}
  \rightsquigarrow
  \big\{\bar{\mathbb B}_w\big\}_{w\in\{0,1\}}
  \quad\text{in }\ell^\infty(\Rex)^2,
\end{equation*}
which proves part (a).

\textbf{Step 2.}
Apply Theorem \ref{theorem:L-weak} to
$\widehat{G}_{Y(w)}$ and $G_{Y(w)}$ for the finite collection
$\{\tau_\ell,\tau_u\}\subset[0,1]$. Part (a) gives, jointly in
$w\in\{0,1\}$,
\begin{equation*}
  \big\{
    \sqrt{n}\big(
      \widehat{G}_{Y(w)}^{\ast}(\tau)-G_{Y(w)}^{\ast}(\tau)
    \big)
  \big\}_{\tau\in\{\tau_\ell,\tau_u\}}
  \rightsquigarrow
  \big\{\bar{\mathbb{Z}}_w(\tau)\big\}_{\tau\in\{\tau_\ell,\tau_u\}}.
\end{equation*}
The endpoint argument is the same as in the proof of Theorem
\ref{theorem:AQTE-asymp}. Since
$\widehat{\theta}(\tau_\ell,\tau_u)$ is a linear combination of these
four conjugate estimators, the continuous mapping theorem proves
part (b).

  \textbf{(c)}
  The exchangeable-bootstrap empirical process gives the
  bootstrap analogue of the expansion in Step 1 when the same weights
  are used to recompute the coefficient estimators and $\bar{b}_n$.
  Applying the delta-method bootstrap of \cite{fang2019inference}
  using the sample $\eta_n$-subdifferentials, followed by the linear
  map in \eqref{eq:aqte-hat}, proves part (c) under the stated rate
  conditions.
\end{proof}
\vspace{0.3cm}

The same result applies when the ReLU regression uses known sampling
weights, under the corresponding weighted moment and rank conditions,
with $(\mathbb B_\gamma,\mathbb B_\delta)$ replaced by the Gaussian
limits of the weighted score processes.

\newpage
\section*{Appendix C: Cluster Sampling}

\setcounter{equation}{0}
\setcounter{theorem}{0}
\setcounter{assumption}{0}

\renewcommand{\theequation}{C.\arabic{equation}}
\renewcommand{\theHequation}{C.\arabic{equation}}
\renewcommand{\thetheorem}{C.\arabic{theorem}}
\renewcommand{\theHtheorem}{C.\arabic{theorem}}
\renewcommand{\theassumption}{C\arabic{assumption}}
\renewcommand{\theHassumption}{C\arabic{assumption}}

This appendix records the extension of the asymptotic theory of
Section \ref{sec:estimation} from i.i.d.\ sampling to cluster
sampling. The argument reduces to existing cluster central limit,
cluster Donsker, and cluster bootstrap results. We omit proofs and
give pointers to the literature.

Suppose the data are organized into $G$ clusters indexed by
$g = 1, \ldots, G$. Cluster $g$ contains $n_{g}$ observations
$\{(Y_{gi}, X_{gi})\}_{i=1}^{n_{g}}$, and the total sample size is
$N := \sum_{g=1}^{G} n_{g}$. Let $\bar{n} := \E[n_g]$ denote the
mean cluster size. Across clusters, the data are i.i.d.
Within a cluster, observations may be arbitrarily dependent. The
estimator of Section \ref{sec:estimation} is unchanged in
form,
\begin{equation*}
  \hat{\beta}(y)
  =
  \bigg( \sum_{g=1}^{G} \sum_{i=1}^{n_{g}} X_{gi} X_{gi}^{\top} \bigg)^{-1}
  \sum_{g=1}^{G} \sum_{i=1}^{n_{g}} X_{gi}\, (y - Y_{gi})_{+}.
\end{equation*}

\vspace{0.3cm}
\begin{assumption}
  \label{assump:cluster}
  The following conditions hold:
  \begin{enumerate}[label=(\alph*), noitemsep, topsep=0pt]
    \item[\textnormal{(a)}] The clusters
      $\{(Y_{g\cdot}, X_{g\cdot})\}_{g=1}^{G}$ are i.i.d.\ across
      $g$.
    \item[\textnormal{(b)}] Cluster sizes $n_{g}$ are uniformly
      bounded by a constant $n_{\max} < \infty$.
    \item[\textnormal{(c)}]
      $\E\big[\sum_{i=1}^{n_g}\{\|X_{gi}\|^4+Y_{gi}^4\}\big]<\infty$.
    \item[\textnormal{(d)}]
      $Q_{X}^{\mathrm{cl}} := \bar{n}^{-1}
      \E\big[\sum_{i=1}^{n_g}X_{gi}X_{gi}^{\top}\big]$ is positive
      definite.
    \item[\textnormal{(e)}] $G \to \infty$.
  \end{enumerate}
\end{assumption}
\vspace{0.3cm}

For each $y \in \R$, let $\beta_0(y)$ denote the pooled population
coefficient
\begin{equation*}
  \beta_0(y)
  :=
  \big(Q_X^{\mathrm{cl}}\big)^{-1}\bar n^{-1}
  \E\bigg[\sum_{i=1}^{n_g}X_{gi}(y-Y_{gi})_+\bigg].
\end{equation*}

For each $y \in \R$, define the cluster score
\begin{equation*}
  S_{g}(y) := \sum_{i=1}^{n_{g}} X_{gi}\, u_{gi}(y),
  \qquad
  u_{gi}(y) := (y - Y_{gi})_{+} - X_{gi}^{\top}\, \beta_{0}(y).
\end{equation*}
The asymptotic covariance kernel under cluster sampling is
\begin{equation*}
  \Sigma^{\mathrm{cl}}(y_{1}, y_{2})
  :=
  \frac{1}{\bar{n}}\,
  \E\big[ S_{g}(y_{1})\, S_{g}(y_{2})^{\top} \big],
\end{equation*}
which decomposes as
\begin{equation}
  \label{eq:cluster-decomp}
  \Sigma^{\mathrm{cl}}(y_{1}, y_{2})
  =
  \Sigma(y_{1}, y_{2})
  +
  \frac{1}{\bar{n}}\,
  \E\bigg[ \sum_{i \ne j} X_{gi} X_{gj}^{\top}\, u_{gi}(y_{1})\, u_{gj}(y_{2}) \bigg],
\end{equation}
where
\begin{equation*}
  \Sigma(y_1,y_2)
  :=
  \frac{1}{\bar n}\,
  \E\bigg[
    \sum_{i=1}^{n_g}
    X_{gi}X_{gi}^{\top}u_{gi}(y_1)u_{gi}(y_2)
  \bigg]
\end{equation*}
is the i.i.d.\ kernel under the size-weighted marginal distribution
of an observation. The second term in \eqref{eq:cluster-decomp}
collects the within-cluster cross-pair contributions. It vanishes
whenever $n_g=1$ almost surely, in which case
$\Sigma^{\mathrm{cl}}=\Sigma$.

\vspace{0.3cm}
\begin{theorem}
  \label{theorem:aym-cluster}
  Suppose Assumption \ref{assump:cluster} holds. Then
  \begin{equation*}
    \sqrt{N}\big( \hat{\beta}(\cdot) - \beta_{0}(\cdot) \big)
    \rightsquigarrow
    \mathbb{B}^{\mathrm{cl}}(\cdot)
    \quad \text{in } \ell^{\infty}(\Rex)^{p},
  \end{equation*}
  where $\mathbb{B}^{\mathrm{cl}}(\cdot)$ is a zero-mean Gaussian
  process with uniformly continuous sample paths and covariance
  function
  $\big(Q_{X}^{\mathrm{cl}}\big)^{-1}\,
  \Sigma^{\mathrm{cl}}(y_{1}, y_{2})\,
  \big(Q_{X}^{\mathrm{cl}}\big)^{-1}$.
\end{theorem}
\vspace{0.3cm}

The proof of Theorem \ref{theorem:aym-cluster} replaces the
i.i.d.\ central limit theorem and the i.i.d.\ Donsker theorem in
the proof of Theorem \ref{theorem:aym} with the cluster
counterparts of \citet[\S 11.4]{kosorok2008introduction}. The
cluster Donsker condition is satisfied under the fourth-moment
condition in
Assumption \ref{assump:cluster}(c) by the same Lipschitz envelope
argument used in the i.i.d.\ case. The decomposition
\eqref{eq:cluster-decomp} follows by direct expansion of the
inner product
$\E[S_{g}(y_{1}) S_{g}(y_{2})^{\top}]$.

The cluster bootstrap implements the exchangeable bootstrap of
\citet{praestgaard1993exchangeably} with cluster-level weights. We
draw a single weight $\xi_{g}$ per cluster $g$ and assign it to every
observation within the cluster.
Equivalently, the resampling unit is the cluster rather than the
individual observation. Consistency of the cluster bootstrap for
the limit law of
$\sqrt{N}(\hat{\beta}(\cdot) - \beta_{0}(\cdot))$ follows from
\citet{sherman1997clustered} and \citet{cheng2013cluster}. The
delta-method bootstrap of Section \ref{sec:estimation} for the
conjugate functional and the bootstrap of
Section \ref{sec:treatment} for the average quantile treatment
effect both extend without modification, because the Hadamard
directional differentiability of the Legendre-Fenchel
transformation does not depend on the resampling unit. The
marginalization result in Theorem
\ref{theorem:marginalized-asymp} likewise extends when the empirical
average in \eqref{eq:marginalized-G} is recomputed using the same
cluster-level weights.

The empirical application of Section \ref{sec:application}
specializes Theorem \ref{theorem:aym-cluster} to the case
$n_{g} \in \{1, 2, 3\}$, with the cluster index $g$ tracking the
household. The bootstrap weights are drawn once per household and
assigned to all individuals in that household, in line with
\citet[\S 4.3]{chernozhukov2020generic}.

\end{document}